\documentclass[12pt]{article}
\usepackage{amssymb}
\usepackage{amsmath}
\usepackage{amsthm}
\usepackage{color}
\usepackage{comment}
\usepackage{enumitem}		
\usepackage{geometry}		
\usepackage{graphicx}

\let\originalleft\left
\let\originalright\right
\renewcommand{\left}{\mathopen{}\mathclose\bgroup\originalleft}
\renewcommand{\right}{\aftergroup\egroup\originalright}

\newcommand{\doroverline}[2]{\overline{#1#2}}
\newcommand{\roverline}[1]{\mathpalette\doroverline{#1}}

\newlist{romanlist}{enumerate}{3}
\setlist[romanlist]{label=\roman*),ref=(\roman*)}

\begin{document}

\newcommand{\cF}{\mathcal{F}}
\newcommand{\cR}{\mathcal{R}}
\newcommand{\cS}{\mathcal{S}}
\newcommand{\cT}{\mathcal{T}}
\newcommand{\cW}{\mathcal{W}}
\newcommand{\ee}{\varepsilon}
\newcommand{\rD}{{\rm D}}
\newcommand{\re}{{\rm e}}

\newtheorem{theorem}{Theorem}[section]
\newtheorem{lemma}[theorem]{Lemma}
\newtheorem{proposition}[theorem]{Proposition}

\theoremstyle{definition}
\newtheorem{definition}{Definition}[section]


\title{
Renormalisation of the two-dimensional border-collision normal form.
}
\author{
I.~Ghosh and
D.J.W.~Simpson\\\\
School of Fundamental Sciences\\
Massey University\\
Palmerston North\\
New Zealand
}
\maketitle


\begin{abstract}

We study the two-dimensional border-collision normal form
(a four-parameter family of continuous, piecewise-linear maps on $\mathbb{R}^2$)
in the robust chaos parameter region of
[S.~Banerjee, J.A.~Yorke, C.~Grebogi, Robust Chaos, {\em Phys. Rev. Lett.}~80(14):3049--3052, 1998].
We use renormalisation to partition this region
by the number of connected components of a chaotic Milnor attractor.
This reveals previously undescribed bifurcation structure in a succinct way.

\end{abstract}

\section{Introduction}
\label{sec:intro}
\setcounter{equation}{0}

Piecewise-linear maps can exhibit complicated dynamics
yet are relatively amenable to an exact analysis.
For this reason they provide a useful tool for us to explore complex aspects of dynamical systems, such as chaos.
They arise as approximations 
to certain types of grazing bifurcations of piecewise-smooth ODE systems \cite{DiBu08},
and are used as mathematical models, particularly in social sciences \cite{PuSu06}.

In this paper we study the family of maps
\begin{equation}
(x,y) \mapsto f_\xi(x,y) = \begin{cases}
\begin{bmatrix} \tau_L x + y + 1 \\ -\delta_L x \end{bmatrix}, & x \le 0, \\
\begin{bmatrix} \tau_R x + y + 1 \\ -\delta_R x \end{bmatrix}, & x \ge 0,
\end{cases}
\label{eq:f}
\end{equation}
where
\begin{equation}
\xi = \left( \tau_L, \delta_L, \tau_R, \delta_R \right).
\label{eq:xi}
\end{equation}
With $(x,y) \in \mathbb{R}^2$ and $\xi \in \mathbb{R}^4$,
this is the two-dimensional border-collision normal form \cite{NuYo92},
except the border-collision bifurcation parameter (often denoted $\mu$) has been scaled to $1$.
It is a normal form in the sense that any continuous, piecewise-linear map with two pieces
for which the image of the switching line intersects the switching line at a unique point that is not a fixed point,
can be transformed to \eqref{eq:f} under an affine change of coordinates, see for instance \cite{Si20e}.
With $\tau_R = -\tau_L$ and $\delta_L = \delta_R$, \eqref{eq:f} reduces to the well-studied Lozi map \cite{Lo78}.

While \eqref{eq:f} appears simple its dynamics can be remarkably rich \cite{BaGr99,Gl16e,Si14,SiMe08b,ZhMo06b}.
In \cite{BaYo98} Banerjee, Yorke, and Grebogi identified an
open parameter region $\Phi_{\rm BYG} \subset \mathbb{R}^4$ (defined below)
throughout which $f_\xi$ has a chaotic attractor, and this was shown formally in \cite{GlSi21}.
Their work popularised the notion that families of piecewise-linear maps typically exhibit chaos in a robust fashion.
This is distinct from families of one-dimensional unimodal maps --- often promoted as a paradigm for chaos
--- that have dense windows of periodicity \cite{GrSw97,Ly97}.
Robust chaos had already been demonstrated by Misiurewicz in the Lozi map \cite{Mi80},
but by studying the border-collision normal form, Banerjee, Yorke, and Grebogi
showed that robust chaos occurs for generic families of piecewise-linear maps.

However, while $f_\xi$ has a chaotic attractor for all $\xi \in \Phi_{\rm BYG}$,
the attractor undergoes bifurcations, or crises \cite{GrOt83}, as the value of $\xi$ is varied within $\Phi_{\rm BYG}$.
The purpose of this paper is to reveal bifurcation structure within $\Phi_{\rm BYG}$
and we achieve this via renormalisation.

Broadly speaking, renormalisation involves showing that, for some member of a family of maps,
a higher iterate or induced map is conjugate to a different member of this family \cite{Ma93}.
By employing this relationship recursively one can obtain far-reaching results.
Renormalisation is central for understanding generic families of one-dimensional maps \cite{CoEc80,DeVa93}.
For instance, Feigenbaum's constant ($4.6692\ldots$)
for the scaling of period-doubling cascades
is the eigenvalue with largest modulus of a fixed point of a renormalisation operator for unimodal maps.

For the one-dimensional analogue of \eqref{eq:f} (skew tent maps)
the bifurcation structure was determined by Ito {\em et.~al.}~\cite{ItTa79b} via renormalisation, see also \cite{VeGl90}.
More recently renormalisation was applied to a two-parameter family
of two-dimensional, piecewise-linear maps in \cite{PuRo18,PuRo19}.
Their results show that for any $n \ge 1$ there exists $\xi \in \mathbb{R}^4$
such that \eqref{eq:f} has $2^n$ coexisting chaotic attractors.

We apply renormalisation to \eqref{eq:f} in the following way.
On the preimage of the closed right half-plane, denoted $\Pi_\xi$,
the second iterate of $f_\xi$ is conjugate to an alternate member of \eqref{eq:f}.
That is, $f_\xi^2$ is conjugate to $f_{g(\xi)}$ for a certain function $g : \mathbb{R}^4 \to \mathbb{R}^4$.
By repeatedly iterating a boundary of $\Phi_{\rm BYG}$ backwards under $g$,
we are able to divide $\Phi_{\rm BYG}$ into regions $\cR_n$, for $n = 0,1,2,\ldots$,
where $f_\xi$ has a chaotic Milnor attractor with $2^n$ connected components.
The regions converge to a fixed point of $g$ as $n \to \infty$.	
The main difficulties we overcome are in analysing the global dynamics of the nonlinear map $g$
and showing that the relevant dynamics of $f_\xi$ occurs entirely within $\Pi_\xi$.


Our main results are presented in \S\ref{sec:results},
see Theorems \ref{th:Rn}--\ref{th:affinelyConjugate}.
Sections \ref{sec:XY}--\ref{sec:mainProof} work toward proofs of these results.
First \S\ref{sec:XY} describes the phase space of \eqref{eq:f},
primarily saddle fixed points and their stable and unstable manifolds.
Then in \S\ref{sec:f2} we consider the second iterate $f_\xi^2$ on $\Pi_\xi$
and construct a conjugacy to $f_{g(\xi)}$.
In \S\ref{sec:phiPsi} we derive geometric properties of the boundaries of $\cR_0$
and in \S\ref{sec:renormalisation} study the dynamics of $g$.

Chaos is proved in the sense of a positive Lyapunov exponent.
This positivity is achieved for all points in the attractor,
including points whose forward orbits intersect the switching line where $f_\xi$ is not differentiable.
This is achieved by using one-sided directional derivatives which are always well-defined in our setting, \S\ref{sec:lyap}.
A recursive application of the renormalisation is performed in \S\ref{sec:mainProof}.
Finally \S\ref{sec:conc} provides a discussion and outlook for future studies. 

\section{Main results}
\label{sec:results}
\setcounter{equation}{0}

In this section we motivate and define the parameter region $\Phi_{\rm BYG}$
and the renormalisation operator $f_\xi \mapsto f_{g(\xi)}$,
then state the main results.
First Theorem \ref{th:Rn} clarifies the geometry of the regions $\cR_n \subset \mathbb{R}^4$.
Next Theorem \ref{th:R0} informs us of the dynamics of $f_\xi$ in $\cR_0$.
Finally Theorem \ref{th:affinelyConjugate} describes the dynamics with $\xi \in \cR_n$ and any value $n \ge 0$
and follows from a recursive application of the renormalisation to Theorem \ref{th:R0}.
Throughout the paper we write
\begin{align}
f_{L,\xi}(x,y) &= \begin{bmatrix} \tau_L x + y + 1 \\ -\delta_L x \end{bmatrix}, &
f_{R,\xi}(x,y) &= \begin{bmatrix} \tau_R x + y + 1 \\ -\delta_R x \end{bmatrix},
\label{eq:fLfR}
\end{align}
for the left and right pieces of \eqref{eq:f}.

\subsection{Two saddle fixed points}
\label{sub:fps}

Consider the parameter region
\begin{equation}
\Phi = \left\{ \xi \in \mathbb{R}^4 \,\big|\,
\tau_L > \delta_L + 1, \,\delta_L > 0, \,\tau_R < -(\delta_R + 1), \,\delta_R > 0 \right\}.
\label{eq:saddleSaddleRegion}
\end{equation}
For any $\xi \in \Phi$, $f_\xi$ has exactly two fixed points.
Specifically
\begin{equation}
Y = \left( \frac{-1}{\tau_L - \delta_L - 1}, \frac{\delta_L}{\tau_L - \delta_L - 1} \right)
\label{eq:Y}
\end{equation}
is a fixed point of $f_{L,\xi}$ and lies in the left half-plane, while
\begin{equation}
X = \left( \frac{-1}{\tau_R - \delta_R - 1}, \frac{\delta_R}{\tau_R - \delta_R - 1} \right)
\label{eq:X}
\end{equation}
is a fixed point of $f_{R,\xi}$ and lies in the right half-plane.

The eigenvalues associated with these points are those of the Jacobian matrices of
$f_{L,\xi}$ and $f_{R,\xi}$:
\begin{align}
A_L(\xi) &= \begin{bmatrix} \tau_L & 1 \\ -\delta_L & 0 \end{bmatrix}, &
A_R(\xi) &= \begin{bmatrix} \tau_R & 1 \\ -\delta_R & 0 \end{bmatrix}.
\label{eq:ALAR}
\end{align}
Notice $\tau_L$ and $\delta_L$ are the trace and determinant of $A_L$;
similarly $\tau_R$ and $\delta_R$ are the trace and determinant of $A_R$.
It follows that $\Phi$ is the set of all parameter combinations for which
$Y$ is a saddle with positive eigenvalues
and $X$ is a saddle with negative eigenvalues.

\subsection{The parameter region $\Phi_{\rm BYG}$}
\label{sub:phiBYG}

For any $\xi \in \Phi$, $X$ and $Y$ have one-dimensional stable and unstable manifolds.
Fig.~\ref{fig:igRN_phasePortrait} illustrates the stable (blue) and unstable (red) manifolds of $Y$.
These intersect if and only if $\phi(\xi) \le 0$, where 
\begin{equation}
\phi(\xi) = \delta_R - (1+\tau_R) \delta_L
+ \frac{1}{2} \big( (1+\tau_R) \tau_L - \tau_R - \delta_L - \delta_R \big)
\left( \tau_L + \sqrt{\tau_L^2 - 4 \delta_L} \right).
\label{eq:phi}
\end{equation}
Equation \eqref{eq:phi} can be derived by directly calculating the first few linear segments
of the stable and unstable manifolds of $Y$
as they emanate from $Y$, see \cite{GlSi21}.
As a bifurcation, $\phi(\xi) = 0$ is a homoclinic corner \cite{Si16b}
and is analogous to a `first' homoclinic tangency for smooth maps \cite{PaTa93}.
Banerjee, Yorke, and Grebogi \cite{BaYo98} observed that an attractor is often destroyed here,
so focussed their attention on the parameter region
\begin{equation}
\Phi_{\rm BYG} = \left\{ \xi \in \Phi \,\big|\, \phi(\xi) > 0 \right\},
\label{eq:BYGRegion}
\end{equation}
where the stable and unstable manifolds of $Y$ do not intersect. 
Indeed for all $\xi \in \Phi_{\rm BYG}$, $f_\xi$ has a trapping region 
and therefore a topological attractor \cite{Gl17}.

\begin{figure}[b!]
\begin{center}
\includegraphics[height=8cm]{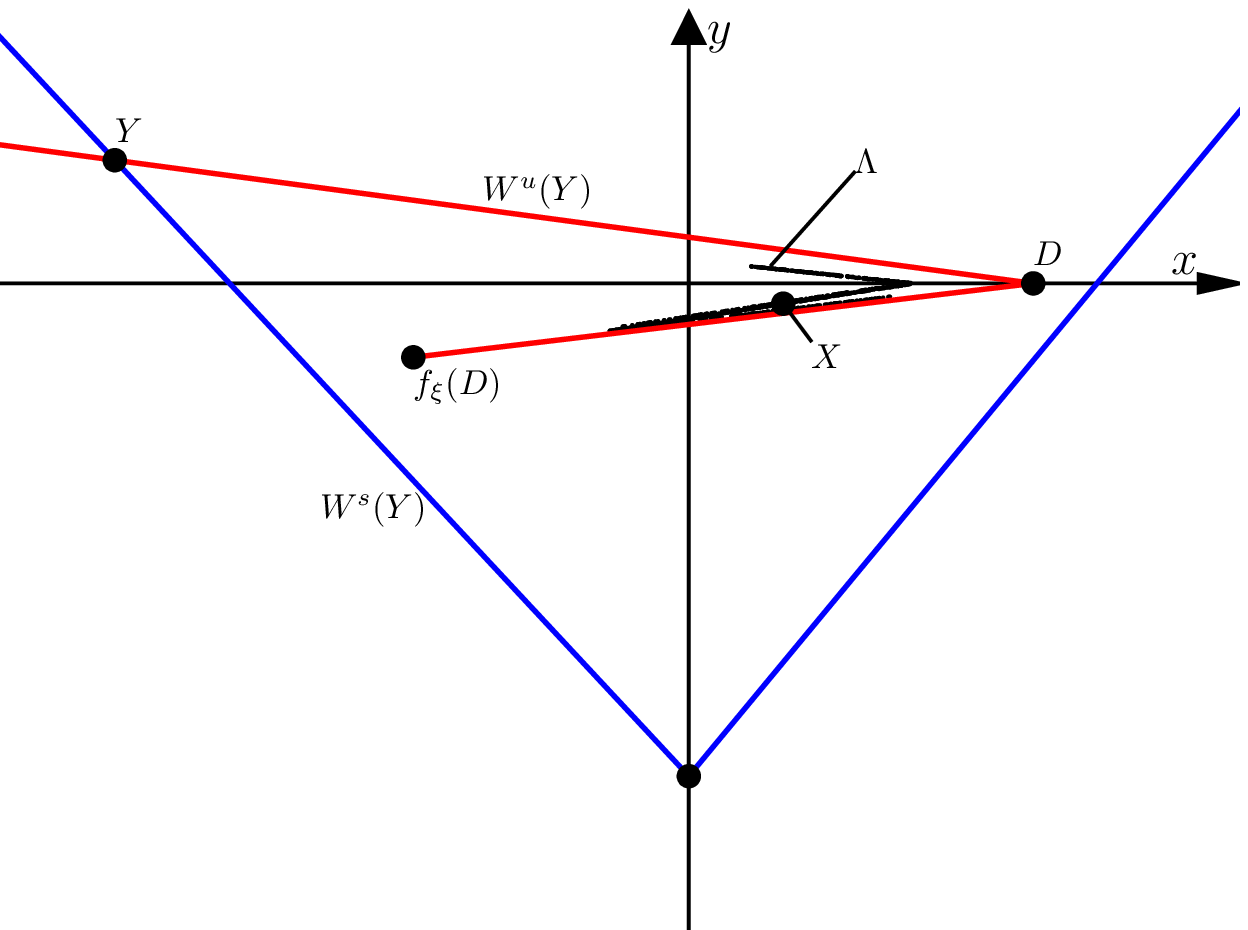}
\caption{
A sketch of the phase space of $f_\xi$ \eqref{eq:f} with $\xi \in \Phi_{\rm BYG}$.
We have shown the fixed points $X$ and $Y$ and the initial parts of $W^s(Y)$ (blue) and $W^u(Y)$ (red)
as they emanate from $Y$ (these manifolds do not intersect when $\phi(\xi) > 0$).
The small black dots show $1000$ iterates of the forward orbit of the origin
after transient dynamics has decayed.
\label{fig:igRN_phasePortrait}
} 
\end{center}
\end{figure}

\subsection{The renormalisation operator}
\label{sub:renormOp}

\begin{figure}[b!]
\begin{center}
\includegraphics[height=5cm]{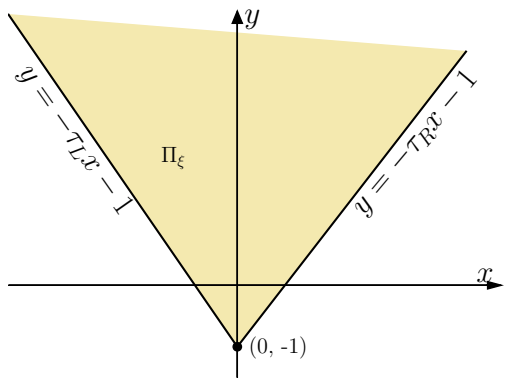}
\caption{
The preimage of the closed right half-plane \eqref{eq:Pi}.
\label{fig:igRN_Pi}
} 
\end{center}
\end{figure}

On $\mathbb{R}^2$ the second iterate $f_\xi^2$ is a continuous, piecewise-linear map with four pieces.
But if we restrict our attention to the set 
\begin{equation}
\Pi_\xi = \left\{ f_\xi^{-1}(x,y) \,\middle|\, x \ge 0 \right\},
\label{eq:Pi}
\end{equation}
then $f_\xi^2$ has only two pieces:
\begin{equation}
f_\xi^2(x,y) =
\begin{cases}
\left( f_{R,\xi} \circ f_{L,\xi} \right)(x,y), & x \le 0, \\
f_{R,\xi}^2(x,y), & x \ge 0.
\end{cases}
\label{eq:f2}
\end{equation}
As shown in Fig.~\ref{fig:igRN_Pi},
the boundary of $\Pi_\xi$ intersects the switching line at $(x,y) = (0,-1)$
and has slope $-\tau_L < 0$ in $x < 0$ and slope $-\tau_R > 0$ in $x > 0$.
For any $\xi \in \Phi$, the map \eqref{eq:f2}
is affinely conjugate to the normal form \eqref{eq:f} (see Proposition \ref{pr:conjugacy}).
This is because the switching line of \eqref{eq:f2} satisfies the non-degeneracy conditions mentioned in \S\ref{sec:intro}.

When the affine transformation to the normal form is applied,
the matrix parts of the pieces of \eqref{eq:f2} undergo a similarity transform,
thus their traces and determinants are not changed.
The matrix part of the $x \le 0$ piece of \eqref{eq:f2}
is $A_R(\xi) A_L(\xi)$, which has trace
$\tau_L \tau_R - \delta_L - \delta_R$ and determinant $\delta_L \delta_R$.
The matrix part of the $x \ge 0$ piece of \eqref{eq:f2}
is $A_R(\xi)^2$, which has trace
$\tau_R^2 - 2 \delta_R$ and determinant $\delta_R^2$.
Hence \eqref{eq:f2} can be transformed to $f_{g(\xi)}$ where
\begin{equation}
g(\xi) = \big( \tau_R^2 - 2 \delta_R, \delta_R^2, \tau_L \tau_R - \delta_L - \delta_R, \delta_L \delta_R \big).
\label{eq:g}
\end{equation}
Notice we are transforming the left piece of \eqref{eq:f2} to the right piece of $f_{g(\xi)}$
and the right piece of \eqref{eq:f2} to the left piece of $f_{g(\xi)}$.
This ensures $g(\xi) \in \Phi$ (see Proposition \ref{pr:PhiForwardInvariant})
so our renormalisation operator $f_\xi \mapsto f_{g(\xi)}$
produces another member of the family \eqref{eq:f} in $\Phi$.
Also observe
\begin{equation}
\xi^* = (1,0,-1,0)
\label{eq:xiStar}
\end{equation}
is a fixed point of $g$ and lies on the boundary of $\Phi$.

\subsection{Division of parameter space}
\label{sub:division}

For all $n \ge 0$ let
\begin{equation}
\zeta_n(\xi) = \phi \big( g^n(\xi) \big).
\label{eq:zetan}
\end{equation}
The surface $\zeta_n(\xi) = 0$ is an $n^{\rm th}$ preimage of $\phi(\xi) = 0$ under $g$.
We now use these surfaces to form the regions
\begin{equation}
\cR_n = \left\{ \xi \in \Phi \,\big|\, \zeta_n(\xi) > 0, \zeta_{n+1}(\xi) \le 0 \right\},
\label{eq:Rn}
\end{equation}
for all $n \ge 0$.
The following result (proved in \S\ref{sub:RnProof}) gives properties of these regions.

\begin{theorem}
The $\cR_n$ are non-empty, mutually disjoint, and converge to $\{ \xi^* \}$ as $n \to \infty$.
Moreover,
\begin{equation}
\Phi_{\rm BYG} \subset \bigcup_{n=0}^\infty \cR_n \,.
\label{eq:RnUnion}
\end{equation}
\label{th:Rn}
\end{theorem}

Being four-dimensional the $\cR_n$ are inherently difficult to visualise.
Fig.~\ref{fig:igRN_parameterSpace} shows two-dimensional cross-sections
obtained by fixing the values of $\delta_L > 0$ and $\delta_R > 0$.
For any such cross-section only finitely many $\cR_n$ are visible because as $n \to \infty$ they converge to $\{ \xi^* \}$
for which $\delta_L = \delta_R = 0$.
Notice $\cR_1$ contains some points that do not belong to $\Phi_{\rm BYG}$.
For this reason the two sets in \eqref{eq:RnUnion} are not equal.

\begin{figure}[b!]
\begin{center}
\setlength{\unitlength}{1cm}
\begin{picture}(16.7,8.5)
\put(-.2,0){\includegraphics[height=8cm]{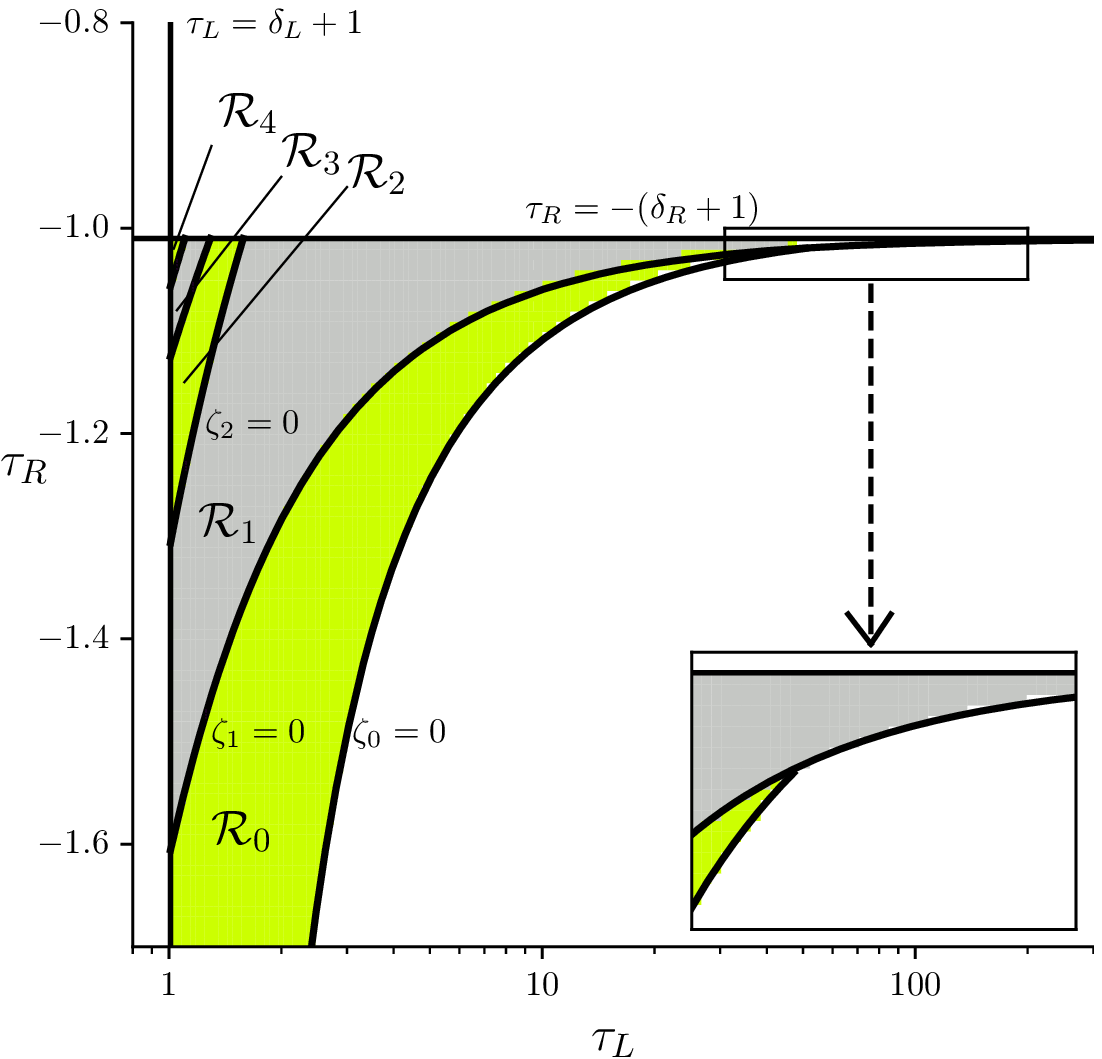}}
\put(8.7,0){\includegraphics[height=8cm]{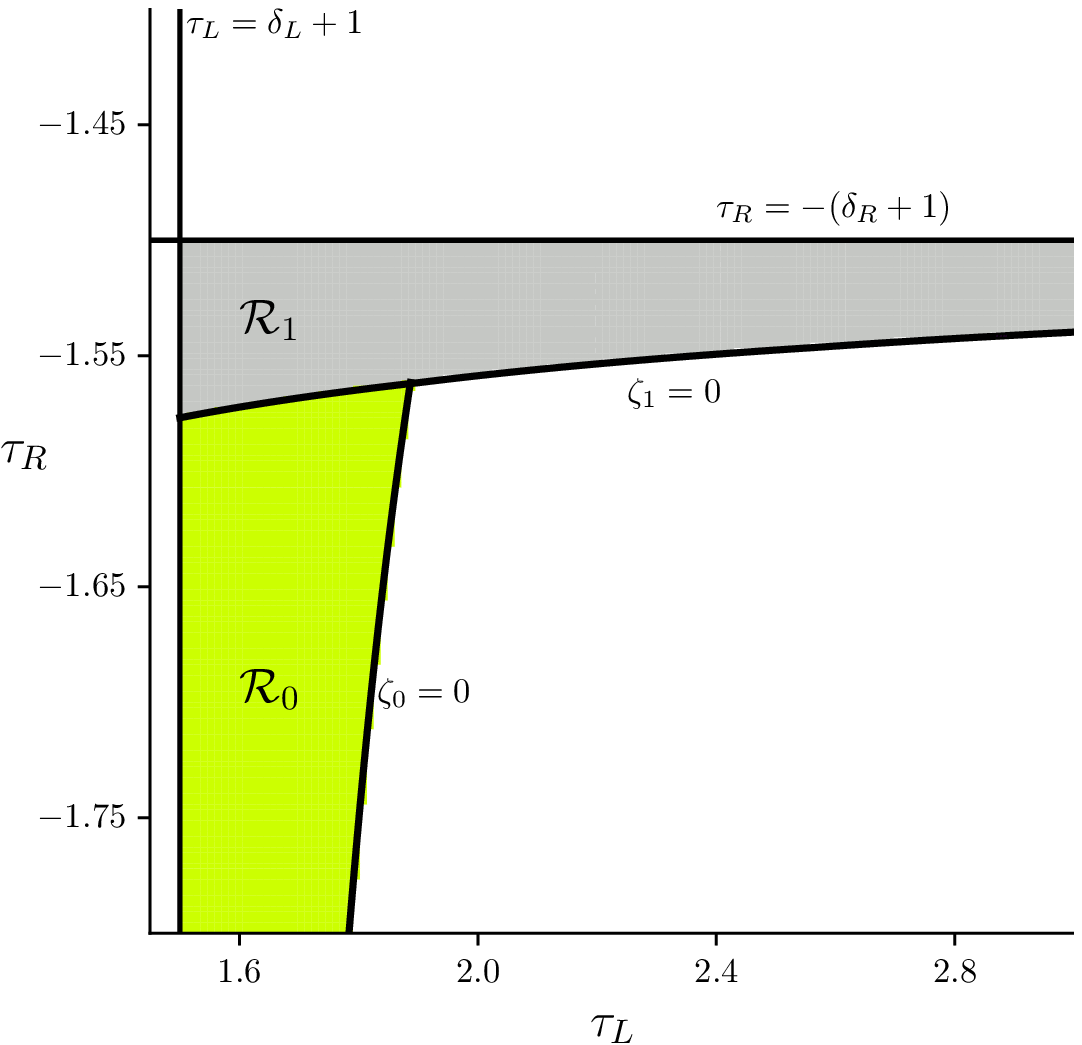}}
\put(3.2,8.2){\small {\bf a)}~~$\delta_L = \delta_R = 0.01$}
\put(12.1,8.2){\small {\bf b)}~~$\delta_L = \delta_R = 0.5$}
\end{picture}
\caption{
Two-dimensional cross-sections of the parameter regions $\cR_n$.
In panel (a) $\cR_n$ is visible for all $n = 0,1,\ldots,4$;
in panel (b) only $\cR_0$ and $\cR_1$ are visible.
In both panels $\Phi_{\rm BYG}$ is the bounded by the vertical line $\tau_L = \delta_L + 1$,
the horizontal line $\tau_R = -\delta_R - 1$,
and the curve $\zeta_0 = 0$.
\label{fig:igRN_parameterSpace}
} 
\end{center}
\end{figure}

\subsection{A chaotic attractor with one connected component}
\label{sub:R0}

The next result shows $f_\xi$ has a chaotic, connected Milnor attractor for all $\xi \in \cR_0$ when $\delta_R < 1$.
This is proved in \S\ref{sub:R0Proof} and based on the results of \cite{GlSi21}.
The attractor is the closure of the unstable manifold of $X$,
\begin{equation}
\Lambda(\xi) = {\rm cl}(W^u(X)).
\label{eq:Lambda}
\end{equation}


\begin{theorem}
For the map $f_\xi$ with any $\xi \in \cR_0$,
\begin{romanlist}
\item
$\Lambda(\xi)$ is bounded, connected, and invariant,
\item
every $z \in \Lambda(\xi)$ has a positive Lyapunov exponent, and
\item
if $\delta_R < 1$ there exists forward invariant $\Delta \subset \mathbb{R}^2$ with non-empty interior such that
\begin{equation}
\bigcap_{n=0}^\infty f_\xi^n(\Delta) = \Lambda(\xi).
\label{eq:LambdaAsInfiniteIntersection}
\end{equation}
\end{romanlist}
\label{th:R0}
\end{theorem}

Lyapunov exponents for \eqref{eq:f} are clarified in \S\ref{sec:lyap}.
Stronger notions of chaos have been obtained on subsets of $\cR_0$, see \cite{Gl17,GlSi21}.
While we have not been able to prove that $\Lambda(\xi)$ is a topological attractor,
\eqref{eq:LambdaAsInfiniteIntersection} shows it contains the $\omega$-limit set of all points in $\Delta$.
The set $\Delta$ has positive Lebesgue measure, thus $\Lambda(\xi)$ is a Milnor attractor \cite{Mi85}.
If $\Delta$ is a trapping region (i.e.~it maps to its interior)
then $\Lambda(\xi)$ is an attracting set by definition \cite{Ro04}.
If $\Delta$ is the trapping region of \cite{GlSi21} (there denoted $\Omega_{\rm trap}$)
then \eqref{eq:LambdaAsInfiniteIntersection} appears to be true for some but not all $\xi \in \cR_0$.
We expect the extra condition $\delta_R < 1$ is unnecessary
but is included in Theorem \ref{th:R0} because our proof utilises an area-contraction argument.

\subsection{A chaotic attractor with many connected components}
\label{sub:mainTheorem}

For any $\xi \in \cR_n$ we have $g^n(\xi) \in \cR_0$ (see Lemma \ref{le:gForwards}),
while Theorem \ref{th:R0} describes the dynamics in $\cR_0$.
Thus by combining the renormalisation with Theorem \ref{th:R0}
we are able to describe the dynamics of $f_\xi$ with $\xi \in \cR_n$.

In view of the way $g$ is constructed, our renormalisation corresponds to the substitution rule
\begin{equation}
(L,R) \mapsto (RR,LR).
\label{eq:substitutionRule}
\end{equation}
The same rule arises in the one-dimensional setting of Ito {\em et.~al.}~\cite{ItTa79b}.
Given a word $\cW$ comprised of $L$'s and $R$'s of length $k$,
let $\cF(\cW)$ be the word of length $2 k$ that results from applying \eqref{eq:substitutionRule}
to every letter in $\cW$.
If an orbit of $f_{g(\xi)}$ has symbolic itinerary $\cW$,
the corresponding orbit of $f_\xi$ has symbolic itinerary $\cF(\cW)$.

The attractor of Theorem \ref{th:R0} is the closure of the unstable manifold of $X$.
Consequently for $\xi \in \cR_n$ the corresponding attractor is the closure of the unstable manifold of
a periodic solution with symbolic itinerary $\cF^n(R)$, see Table \ref{tb:cF}.

\begin{table}[b!]
\begin{center}
\begin{tabular}{c|c}
$n$ & $\cF^n(R)$ \\
\hline
$0$ & R \\
$1$ & LR \\
$2$ & RRLR \\
$3$ & LRLRRRLR \\
$4$ & RRLRRRLRLRLRRRLR
\end{tabular}
\end{center}
\caption{
The first few words in the sequence generated by repeatedly applying
the symbolic substitution rule \eqref{eq:substitutionRule} to $R$.
\label{tb:cF}
}
\end{table}

\begin{theorem}
Let $n \ge 0$ and $\xi \in \cR_n$.
Then $g^n(\xi) \in \cR_0$ and
there exist mutually disjoint sets $S_0,S_1,\ldots,S_{2^n-1} \subset \mathbb{R}^2$ such that
$f_\xi(S_i) = S_{(i+1) \,{\rm mod}\, 2^n}$ and
\begin{equation}
f_\xi^{2^n} \big|_{S_i} ~\text{is affinely conjugate to}~ f_{g^n(\xi)} \big|_{\Lambda(g^n(\xi))}
\label{eq:affinelyConjugate}
\end{equation}
for each $i \in \{ 0,1,\ldots,2^n-1 \}$.
Moreover,
\begin{equation}
\bigcup_{i=0}^{2^n-1} S_i = {\rm cl} \left( W^u \left( \gamma_n \right) \right),
\label{eq:Siunion}
\end{equation}
where $\gamma_n$ is a saddle-type periodic solution of $f_\xi$ with symbolic itinerary $\cF^n(R)$.
\label{th:affinelyConjugate}
\end{theorem}

Numerical explorations suggest that \eqref{eq:Siunion} is the unique attractor of \eqref{eq:f} for any $\xi \in \cR_n$.
Theorem \ref{th:affinelyConjugate} tells us it has $2^n$ connected components and is the closure
of the unstable manifold of a saddle-type period-$2^n$ solution.
Each component $S_i$ is invariant under $2^n$ iterations of $f_\xi$.
Equation \eqref{eq:affinelyConjugate} tells us that the dynamics of $f_\xi^{2^n}$ on $S_i$ 
is equivalent (under an affine coordinate change) to that of $f_{g^n(\xi)}$ on $\Lambda(g^n(\xi))$.
Since $g^n(\xi) \in \cR_0$, the properties listed in Theorem \ref{th:R0} apply to $f_\xi^{2^n}$ on $S_i$.
Thus \eqref{eq:Siunion} is a chaotic Milnor attractor of $f_\xi$.

As an example, consider $f_\xi$ with
\begin{equation}
\xi_{\rm ex} = (1.15,0.01,-1.12,0.01) \in \cR_2 \,.
\label{eq:xiExample}
\end{equation}
Fig.~\ref{fig:igRN_phasePortrait2}-a shows $1000$ points
of the forward orbit of the origin after transient behaviour has decayed.
As expected these points appear to converge to a chaotic attractor with four connected components.
By Theorem \ref{th:affinelyConjugate} each component is affinely conjugate to $\Lambda(g^2(\xi))$
which is approximated in Fig.~\ref{fig:igRN_phasePortrait2}-b by again iterating the origin.
The set $\Lambda(g^2(\xi))$ has a complicated branched structure
but this is not visible in Fig.~\ref{fig:igRN_phasePortrait2}-b
because the determinants are extremely small.

\begin{figure}[b!]
\begin{center}
\setlength{\unitlength}{1cm}
\begin{picture}(16.7,8.5)
\put(.2,0){\includegraphics[height=8cm]{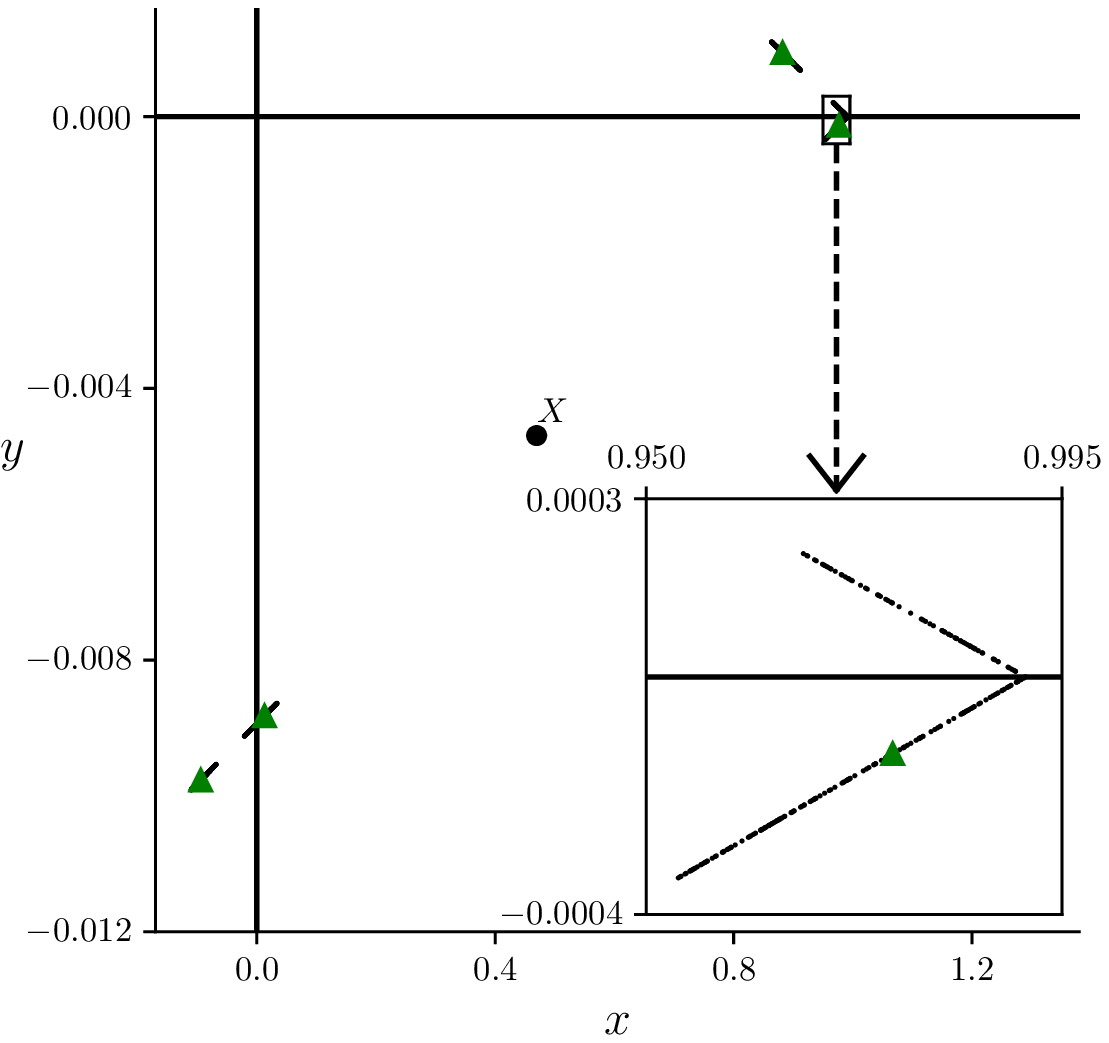}}
\put(9.1,0){\includegraphics[height=8cm]{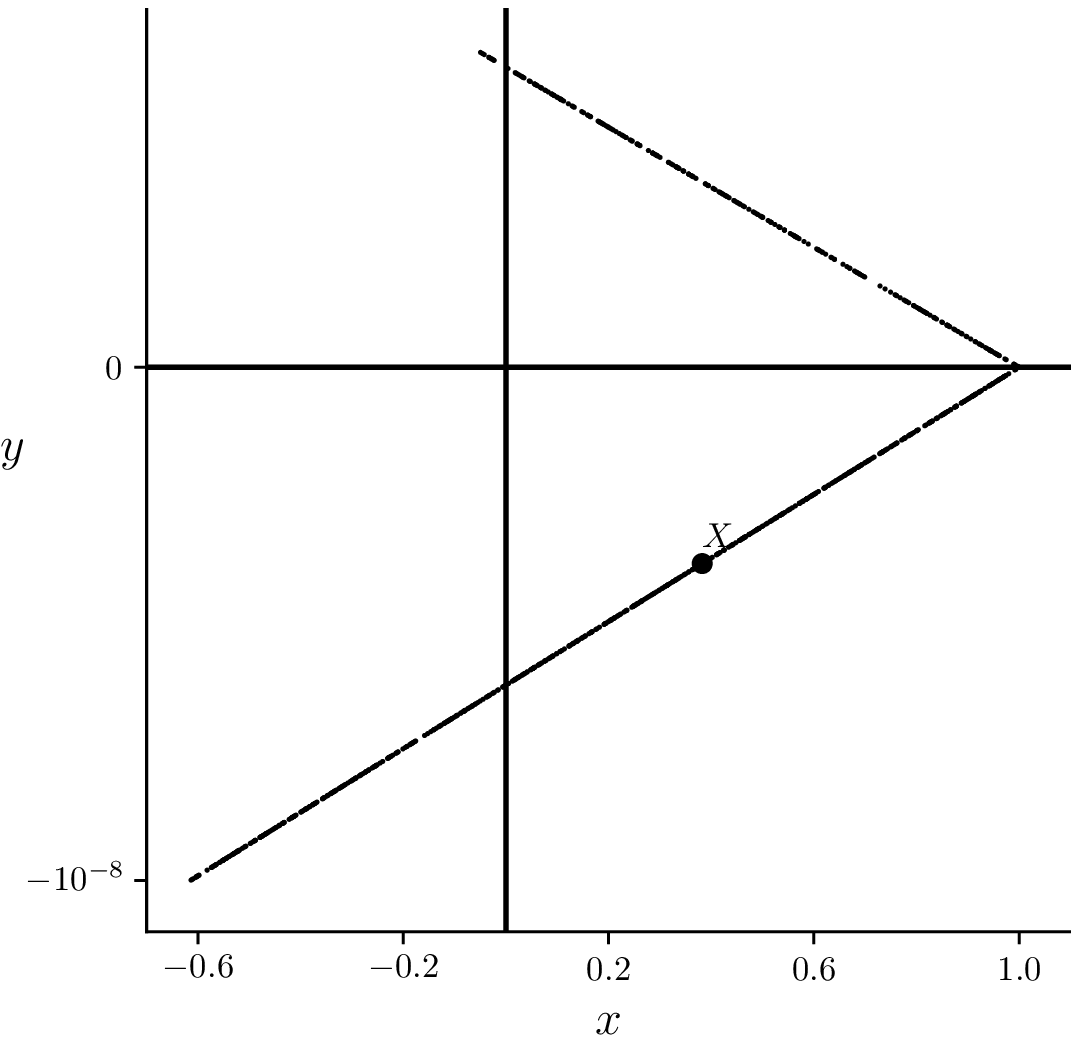}}
\put(3.6,8.2){\small {\bf a)}~~$\xi = \xi_{\rm ex}$}
\put(12.5,8.2){\small {\bf b)}~~$\xi = g^2(\xi_{\rm ex})$}
\end{picture}
\caption{
Numerically computed attractors of $f_\xi$ with $\xi = \xi_{\rm ex}$, \eqref{eq:xiExample}, in panel (a),
and $\xi = g^2(\xi_{\rm ex})$ in panel (b).
In panel (a) the four small triangles are the points of a periodic solution
with symbolic itinerary $\cF^2(R) = RRLR$.
\label{fig:igRN_phasePortrait2}
} 
\end{center}
\end{figure}


\section{The stable and unstable manifolds of the fixed points}
\label{sec:XY}
\setcounter{equation}{0}

In this section we discuss the stable and unstable manifolds of the saddle fixed points $X$ and $Y$.
Here and throughout the paper
\begin{equation}
0 < \lambda_L^s < 1 < \lambda_L^u
\label{eq:eigsAL}
\end{equation}
denote the eigenvalues of $A_L$, and
\begin{equation}
\lambda_R^u < -1 < \lambda_R^s < 0
\label{eq:eigsAR}
\end{equation}
denote the eigenvalues of $A_R$.
These are functions of $\xi$ and assume $\xi \in \Phi$.

\subsection{Stable and unstable manifolds of piecewise-linear maps}
\label{sub:P}

Let $P$ be one of the saddle fixed points $X$ or $Y$.
The stable manifold of $P$ is defined as
\begin{equation}
W^s(P) = \left\{ z \in \mathbb{R}^2 \setminus \{ P \} \,\big|\, f_\xi^n(z) \to P ~\text{as}~ n \to \infty \right\}.
\label{eq:WsDefn}
\end{equation}
For all $\xi \in \Phi$ the map $f_\xi$ is invertible so the unstable manifold of $P$ is defined analogously as
\begin{equation}
W^u(P) = \left\{ z \in \mathbb{R}^2 \setminus \{ P \} \,\big|\, f_\xi^{-n}(z) \to P ~\text{as}~ n \to \infty \right\}.
\label{eq:WuDefn}
\end{equation}
Since $P$ is a saddle, $W^s(P)$ and $W^u(P)$ are one-dimensional.
As with smooth maps, from $P$ they emanate tangent
to the stable and unstable subspaces $E^s(P)$ and $E^u(P)$.
These subspaces are the lines through $P$ with directions given by the eigenvectors of $\rD f_\xi (P)$.
But since $f_\xi$ is piecewise-linear,
$W^s(P)$ and $W^u(P)$ in fact {\em coincide} with $E^s(P)$ and $E^u(P)$ in a neighbourhood of $P$.
Globally they have a piecewise-linear structure: 
$W^s(P)$ has kinks on the switching line $x=0$ and on the backward orbits of these points;
$W^u(P)$ has kinks on the image of switching line, $y=0$, and on the forward orbits of these points.

In the remainder of this section we reproduce the geometric constructions
of \cite{GlSi21} that will be needed below.

\subsection{The stable and unstable manifolds of $Y$}
\label{sub:Y}

\begin{figure}[b!]
\begin{center}
\includegraphics[height=8cm]{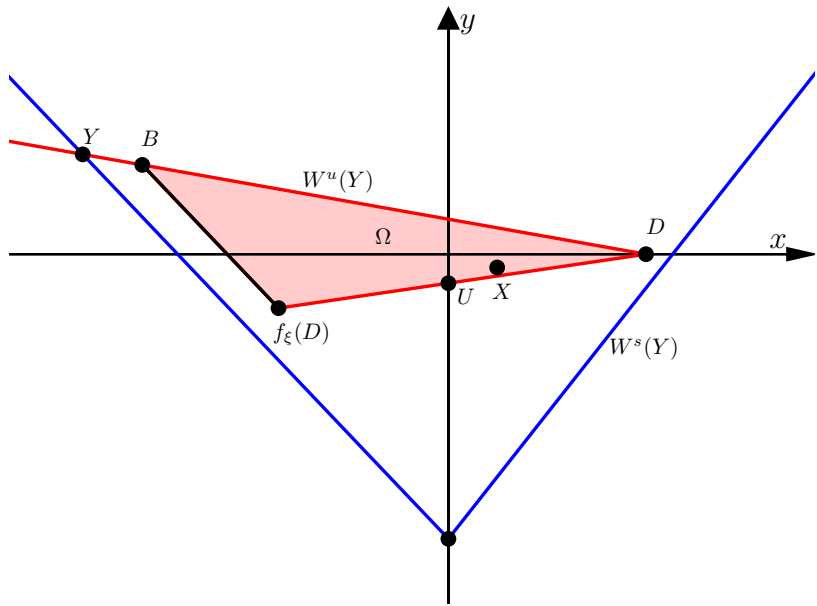}
\caption{
A sketch of the phase space of $f_\xi$ with $\xi \in \Phi_{\rm BYG}$.
The triangle $\Omega(\xi)$ is shaded.
\label{fig:igRN_Y}
} 
\end{center}
\end{figure}

Since the eigenvalues of $A_L$ are positive,
$W^s(Y)$ and $W^u(Y)$ each have two dynamically independent branches.
Let $D$ denote the first kink of the right branch of $W^u(Y)$ as we follow it outwards from $Y$,
see Fig.~\ref{fig:igRN_Y}.
Notice $D$ is the intersection of $E^u(Y)$ with $y=0$.
Now let $B$ denote the intersection of $E^u(Y)$ with the line through $f_\xi(D)$ and parallel to $E^s(Y)$.
Then let $\Omega(\xi)$ be the closed compact triangle with vertices $D$, $f_\xi(D)$, and $B$.

The following result says $\Omega(\xi)$ is forward invariant under $f_\xi$.
This was proved in \cite{GlSi21} by direct calculations.
The key observation is that $f_\xi(D)$ lies to the right of $E^s(Y)$ because $\phi(\xi) > 0$.

\begin{proposition}
For any $\xi \in \Phi_{\rm BYG}$, $f_\xi \left( \Omega(\xi) \right) \subset \Omega(\xi)$.
\label{pr:Omega}
\end{proposition}

The next result tells us that the attractor of Theorem \ref{th:R0} is contained in $\Omega(\xi)$.

\begin{lemma}
For any $\xi \in \Phi_{\rm BYG}$, $\Lambda(\xi) \subset \Omega(\xi)$.
\label{le:LambdaInOmega}
\end{lemma}

\begin{proof}
Since $\Omega(\xi)$ is forward invariant we only need to show $X \in \Omega(\xi)$.
By direct calculations 
we find that the line through $D$ and $f_\xi(D)$ is $y = \ell(x)$ where
\begin{equation}
\ell(x) = \frac{\delta_R}{\lambda_L^s - \tau_R} \left( x - \frac{1}{1 - \lambda_L^s} \right).
\nonumber
\end{equation}
From \eqref{eq:X} we obtain, after much simplification,
\begin{equation}
X_2 - \ell(X_1) = \frac{\delta_R \left( \lambda_L^{s^2} - \tau_R \lambda_L^s + \delta_R \right)}
{(\delta_R + 1 - \tau_R) \left( \lambda_L^s - \tau_R \right) \left( 1 - \lambda_L^s \right)}.
\nonumber
\end{equation}
In view of \eqref{eq:saddleSaddleRegion} and \eqref{eq:eigsAL},
each factor in this expression is positive,
thus $X$ lies above the line through $D$ and $f_\xi(D)$.
Also $X_1 > 0$ and $X_2 < 0$, thus $X \in \Omega(\xi)$ as required.
\end{proof}

\subsection{The stable and unstable manifolds of $X$}
\label{sub:X}

\begin{figure}[b!]
\begin{center}
\setlength{\unitlength}{1cm}
\begin{picture}(16.5,8.6)
\put(0,0){\includegraphics[height=8cm]{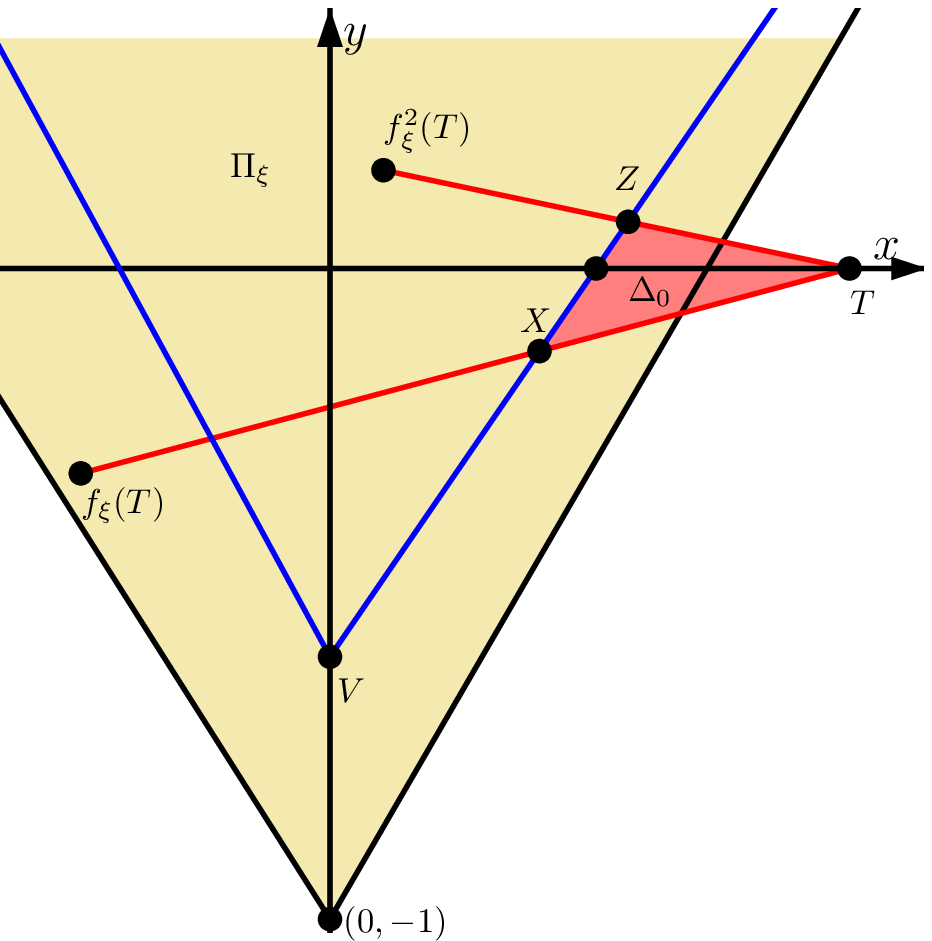}}
\put(8.5,0){\includegraphics[height=8cm]{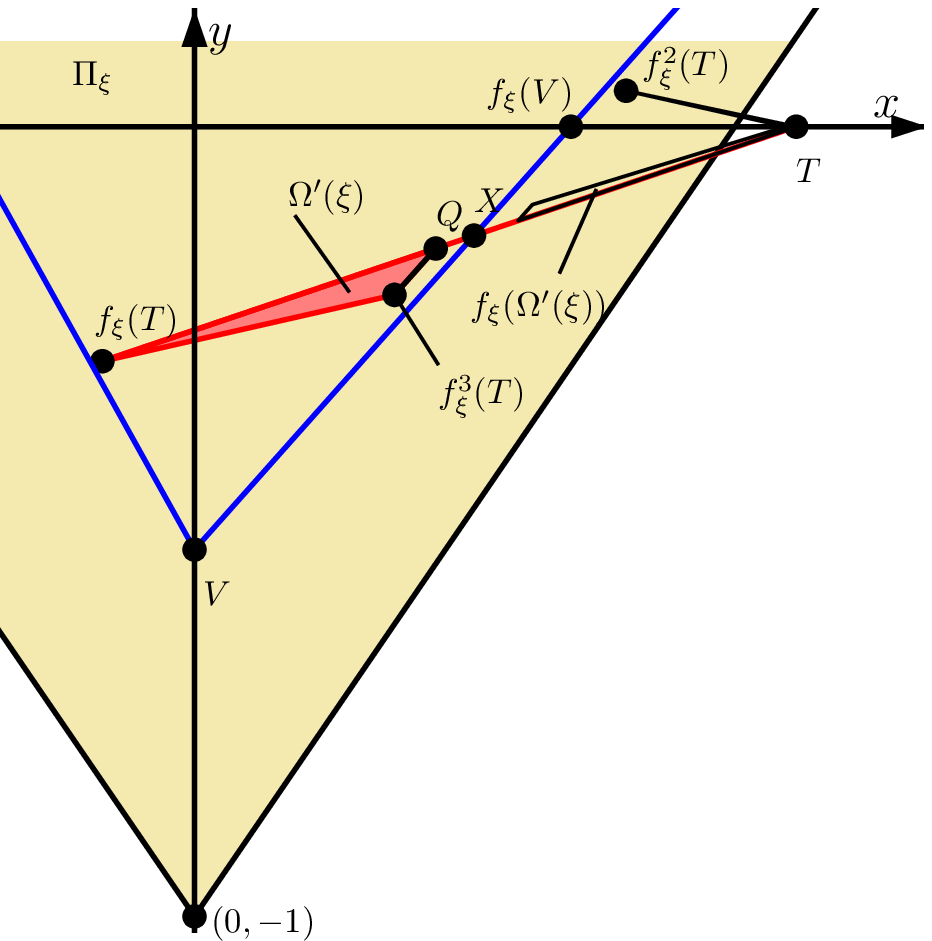}}
\put(3.4,8.3){\small {\bf a)}~~$\xi \in \cR_0$}
\put(11.5,8.3){\small {\bf b)}~~$\xi \in \cR_n \,, n \ge 1$}
\end{picture}
\caption{
Sketches of phase space with $\xi \in \cR_0$ in panel (a) and $\xi \in \cR_n$ with $n \ge 1$ in panel (b).
The set $\Delta_0$ in panel (a) is introduced in \S\ref{sub:R0Proof}.
The set $\Omega'$ in panel (b) is introduced in \S\ref{sub:OmegaPrime}.
\label{fig:igRN_X}
} 
\end{center}
\end{figure}

Since the eigenvalues of $A_R$ are negative,
$W^s(X)$ and $W^u(X)$ each have one dynamically independent branch.
Let $T$ denote the intersection of $E^u(X)$ with $y=0$
and let $V$ denote the intersection of $E^s(X)$ with $x=0$,
see Fig.~\ref{fig:igRN_X}.
It is easily shown that
\begin{equation}
T = \left( \frac{1}{1 - \lambda_R^s}, 0 \right).
\label{eq:T}
\end{equation}
If $f_\xi^2(T)$ lies to the left of $E^s(X)$, as in Fig.~\ref{fig:igRN_X}-a,
then $W^s(X)$ and $W^u(X)$ intersect transversely.
If $f_\xi^2(T)$ lies to the right of $E^s(X)$, as in Fig.~\ref{fig:igRN_X}-b,
then $W^s(X)$ and $W^u(X)$ have no intersection.
The following result was obtained in \cite{Gl17} by calculating $f_\xi^2(T)$ explicitly.

\begin{proposition}
For any $\xi \in \Phi$, $f_\xi^2(T)$ lies to the left of $E^s(X)$ if and only if $\psi(\xi) > 0$, where
\begin{align}
\psi(\xi) = (\tau_L \tau_R - \delta_R) \lambda_R^u
+ \left( \frac{\delta_L}{\delta_R} + \delta_L - 1 \right) \lambda_R^s
- \tau_L (1 + \delta_R) + \tau_R (1 - \delta_L).
\label{eq:psi}
\end{align}
\label{pr:psi}
\end{proposition}

As a bifurcation, $\psi(\xi) = 0$ is a homoclinic corner for the fixed point $X$.
This is analogous to the surface $\phi(\xi) = 0$ for the fixed point $Y$ as discussed in \S\ref{sub:phiBYG}.

\section{The second iterate of $f_\xi$}
\label{sec:f2}
\setcounter{equation}{0}

As discussed in \S\ref{sub:renormOp},
on $\Pi_\xi$ the second iterate of $f_\xi$
is a continuous, piecewise-linear map with two pieces, \eqref{eq:f2}.
Next in \S\ref{sub:conjugacy} we provide the affine transformation that converts \eqref{eq:f2} to the normal form \eqref{eq:f}.
Then in \S\ref{sub:psiAgain} we show that the bifurcation surface $\psi(\xi) = 0$ of the previous section
is in fact identical to $\zeta_1(\xi) = \phi(g(\xi)) = 0$.


\subsection{A transformation to the normal form}
\label{sub:conjugacy}

Any continuous, two-piece, piecewise-linear map on $\mathbb{R}^2$
for which the image of the switching line intersects the switching line at a unique point that is not a fixed point
can be transformed to \eqref{eq:f} under an affine coordinate transformation.
The required transformation is described in the original work \cite{NuYo92}.
For the generalisation to $n$ dimensions refer to \cite{Si16}.

The switching line of \eqref{eq:f2} satisfies this condition for any $\xi \in \Phi$.
As clarified by Proposition \ref{pr:conjugacy}, the required coordinate transformation is
\begin{equation}
h_\xi(x,y) = \frac{1}{\tau_R + \delta_R + 1} \begin{bmatrix} x \\ \delta_R x + \tau_R y - \delta_R \end{bmatrix}.
\label{eq:h}
\end{equation}

\begin{proposition}
For any $\xi \in \Phi$,
\begin{equation}
f_\xi^2 = h_\xi^{-1} \circ f_{g(\xi)} \circ h_\xi \,,
\label{eq:conjugacy}
\end{equation}
on $\Pi_\xi$.
\label{pr:conjugacy}
\end{proposition}

\begin{proof}
By directly composing \eqref{eq:fLfR} and \eqref{eq:h} we obtain
\begin{equation}
h_\xi \circ f_\xi^2 = \begin{cases}
\dfrac{1}{\tau_R + \delta_R + 1} \begin{bmatrix}
\left( \tau_R^2 - \delta_R \right) x + \tau_R y + \tau_R + 1 \\
-\delta_R^2 x
\end{bmatrix}, & x \le 0, \\
\dfrac{1}{\tau_R + \delta_R + 1} \begin{bmatrix}
\left( \tau_L \tau_R - \delta_L \right) x + \tau_R y + \tau_R + 1 \\
-\delta_L \delta_R x
\end{bmatrix}, & x \ge 0,
\end{cases}
\nonumber
\end{equation}
and it is readily seen that $f_{g(\xi)} \circ h_\xi$ produces the same expression.
\end{proof}

Write $(\tilde{x},\tilde{y}) = h_\xi(x,y)$.
Notice that $x$ and $\tilde{x}$ have opposite signs, i.e.
\begin{equation}
{\rm sgn}(x) = -{\rm sgn}(\tilde{x}).
\label{eq:oppositeSigns}
\end{equation}
This is because $\tau_R + \delta_R + 1 < 0$ by \eqref{eq:saddleSaddleRegion}.
Thus the left piece of $f_{g(\xi)}$
corresponds to the right piece of $f_\xi^2$ in \eqref{eq:f2},
and this is consistent with how $g$ was introduced in \S\ref{sub:renormOp}.

\subsection{A reinterpretation of $\psi$}
\label{sub:psiAgain}

In \S\ref{sub:X} we saw that the fixed point $X$ of $f_\xi$ has a homoclinic corner when $\psi(\xi) = 0$.
The same is true for $f_\xi^2$: its fixed point $X$ has a homoclinic corner when $\psi(\xi) = 0$.
Notice $X$ is a fixed point of $f_{R,\xi}^2$, which is transformed under \eqref{eq:conjugacy} to $f_{L,g(\xi)}$, which has the fixed point $Y$.
Thus, while the stable and unstable manifolds of $X$ lie in $\Pi_\xi$,
they transform to the stable and unstable manifolds of $Y$ for $f_{g(\xi)}$.
The latter manifolds have a homoclinic corner when $\phi(g(\xi)) = 0$,
which suggests that $\psi(\xi) = 0$ and $\phi(g(\xi)) = 0$ are the same surface.
The following result tells us that this is indeed the case.

\begin{lemma}
For any $\xi \in \Phi$,
\begin{equation}
\phi(g(\xi)) = \tau_R \lambda_R^{u^2} \psi(\xi).
\label{eq:psi2}
\end{equation}
\label{le:psizeta1Relationship}
\end{lemma}

\begin{proof}
Equation \eqref{eq:phi} can be written as
\begin{equation}
\phi(\xi) = (1+\tau_R) \lambda_L^{u^2} - (\tau_R + \delta_L + \delta_R) \lambda_L^u + \delta_R \,.
\label{eq:phi2}
\end{equation}
To evaluate $\phi(g(\xi))$, in \eqref{eq:phi2} we replace $\delta_L$ with $\delta_R^2$, $\delta_R$ with $\delta_L \delta_R$,
and $\tau_R$ with $\tau_L \tau_R - \delta_L - \delta_R$, see \eqref{eq:g}.
Also we replace $\lambda_L^u$ with $\lambda_R^{u^2}$
because $\lambda_R^{u^2}$ is the unstable eigenvalue of $A_R^2$
(which has trace and determinant given by the first two components of \eqref{eq:g}).
It is a simple (though tedious) exercise to show that upon performing these substitutions and simplifying we obtain
$\tau_R \lambda_R^{u^2} \psi(\xi)$.
\end{proof}

\section{The geometry of the boundary of $\cR_0$}
\label{sec:phiPsi}
\setcounter{equation}{0}


The region $\cR_0 \subset \mathbb{R}^4$ is bounded by $\zeta_0(\xi) = \phi(\xi) = 0$,
$\zeta_1(\xi) = \phi(g(\xi)) = 0$,
and the hyperplanes specified in \eqref{eq:saddleSaddleRegion}.
Since parameter space is four-dimensional these are difficult to visualise.
We can benefit from the fact that the $\delta_L$ and $\delta_R$ components of
$g$ are decoupled from $\tau_L$ and $\tau_R$.
Thus two-dimensional slices
\begin{equation}
\Phi_{\rm slice}(\delta_L,\delta_R) = \left\{ (\tau_L,\tau_R) \,\middle|\,
\tau_L > \delta_L + 1, \tau_R < -\delta_R-1 \right\},
\label{eq:Phislice}
\end{equation}
defined by fixing the values of $\delta_L$ and $\delta_R$,
map to one another under $g$.
In any such slice $\zeta_0(\xi) = 0$ and $\zeta_1(\xi) = 0$ are curves.
In this section we show that for any values $0 < \delta_L < 1$ and $0 < \delta_R < 1$,
these curves have the geometry shown in Fig.~\ref{fig:igRN_phipsi}.

\begin{figure}[b!]
\begin{center}
\includegraphics[height=10cm]{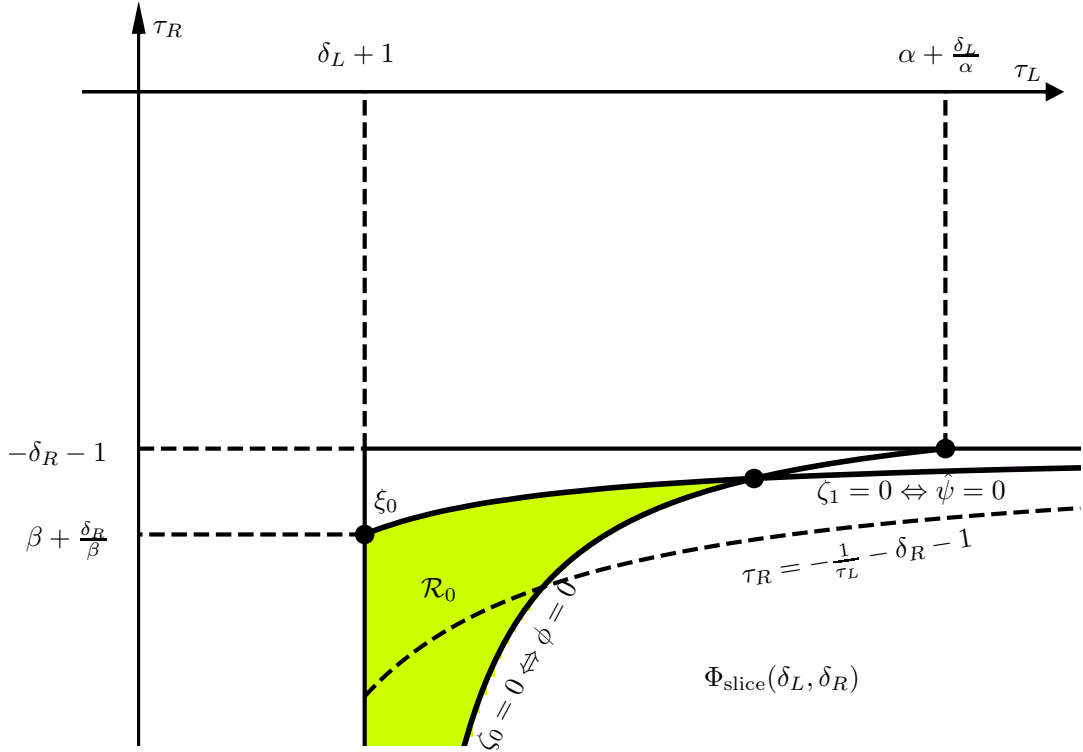}
\caption{
A sketch of $\zeta_0(\xi) = 0$ and $\zeta_1(\xi) = 0$
(equivalently $\phi(\xi) = 0$ and $\hat{\psi}(\xi) = 0$)
in $\Phi_{\rm slice}(\delta_L,\delta_R)$ with $0 < \delta_L < 1$ and $0 < \delta_R < 1$.
The curve $\tau_R = -\frac{1}{\tau_L} - \delta_R - 1$ is shown dashed.
\label{fig:igRN_phipsi}
} 
\end{center}
\end{figure}

Observe $\zeta_0(\xi) = 0$ is the same as 
$\phi(\xi) = 0$, while, by Lemma \ref{le:psizeta1Relationship},
$\zeta_1(\xi) = 0$ is the same as $\psi(\xi) = 0$.
However, we find the function
\begin{equation}
\hat{\psi}(\xi) = \lambda_R^u \psi(\xi),
\label{eq:psiHatDefn}
\end{equation}
easier to work with $\psi(\xi)$.
By \eqref{eq:psi2} the sign of $\hat{\psi}(\xi)$ is the same as that of $\zeta_1(\xi)$.
From \eqref{eq:psi} we obtain
\begin{equation}
\hat{\psi}(\xi) = -\delta_L \left( \lambda_R^{u^2} - 1 \right) + \lambda_R^u \left( \lambda_R^{u^2} - 1 \right) \tau_L
+ (1 - \delta_R) \lambda_R^{u^2} \,.
\label{eq:psiHat}
\end{equation}

The remainder of this section is organised as follows.
First in \S\ref{sub:phi} we study the curve $\phi(\xi) = 0$.
We then derive analogous properties for $\hat{\psi}(\xi) = 0$ and obtain some additional bounds, \S\ref{sub:psi}.
Lastly we show these curves intersect at a unique point in $\Phi_{\rm slice}$, \S\ref{sub:phiandpsi}.

\subsection{The curve $\phi(\xi) = 0$}
\label{sub:phi}

We first show the curve $\phi(\xi) = 0$ does not exist
in $\Phi_{\rm slice}(\delta_L,\delta_R)$ if $\delta_L \ge 1$.

\begin{lemma}
Let $\xi \in \Phi$.
If $\delta_L \ge 1$ then $\phi(\xi) < 0$.
\label{le:deltaLge1}
\end{lemma}

\begin{proof}
We can rearrange \eqref{eq:phi2} as
\begin{equation}
\phi(\xi) = (\tau_R + \delta_R + 1) \lambda_L^u \left( \lambda_L^u - 1 \right)
- \delta_R \left( \lambda_L^{u^2} - 1 \right)
+ (1 - \delta_L) \lambda_L^u \,.
\label{eq:phi3}
\end{equation}
By inspection the first two terms in \eqref{eq:phi3} are negative
and if $\delta_L \ge 1$ then the last term is less than or equal to zero.
\end{proof}

The next result shows that $\phi(\xi) = 0$ appears roughly as in Fig.~\ref{fig:igRN_phipsi}.

\begin{proposition}
Let $0 < \delta_L < 1$ and $\delta_R > 0$.
There exists a unique $C^\infty$ function $G : (-\infty,-\delta_R-1] \to (\delta_L+1,\infty)$
such that
\begin{equation}
\phi \big( G(\tau_R), \delta_L, \tau_R, \delta_R \big) = 0,
\label{eq:phiZeroCurve}
\end{equation}
for all $\tau_R \in (-\infty,-\delta_R-1]$.
Moreover,
$G$ is strictly increasing,
$G(\tau_R) \to \delta_L+1$ as $\tau_R \to -\infty$,
and $G(-\delta_R-1) = \alpha + \frac{\delta_L}{\alpha}$ where $\alpha \in \mathbb{R}$ is the largest solution to
\begin{equation}
-\delta_R \alpha^2 + (1-\delta_L) \alpha + \delta_R = 0\,.
\label{eq:phiTopIntersection}
\end{equation}
\label{pr:phiZeroCurve}
\end{proposition}

\begin{proof}
First fix $\tau_R \le -\delta_R-1$.
With $\tau_L = \delta_L + 1$ we have $\lambda_L^u = 1$ and so \eqref{eq:phi2} simplifies to
$\phi(\xi) = 1 - \delta_L > 0$.
As $\tau_L \to \infty$ we have $\lambda_L^u \to \infty$ and so $\phi(\xi) \to -\infty$
(because the $\lambda_L^{u^2}$-coefficient in \eqref{eq:phi2} is negative).
Thus by the intermediate value theorem
there exists $\tau_L = G(\tau_R) > \delta_L + 1$ satisfying \eqref{eq:phiZeroCurve}.

To demonstrate the uniqueness of $G$ we differentiate \eqref{eq:phi2} to obtain
\begin{equation}
\frac{\partial \phi}{\partial \tau_L} =
\left( 2 (1 + \tau_R) \lambda_L^u - (\tau_R + \delta_L + \delta_R) \right) \frac{\partial \lambda_L^u}{\partial \tau_L}.
\label{eq:phiZeroCurveProof10}
\end{equation}
It is a simple exercise to show that $\frac{\partial \lambda_L^u}{\partial \tau_L} = \frac{\lambda_L^u}{\lambda_L^u - \lambda_L^s}$.
Also if $\phi = 0$ then by \eqref{eq:phi2}
we can replace $(\tau_R + \delta_L + \delta_R)$ in \eqref{eq:phiZeroCurveProof10}
with $\frac{\delta_R}{\lambda_L^u} + (1+\tau_R) \lambda_L^u$ to obtain
\begin{equation}
\frac{\partial \phi}{\partial \tau_L} \bigg|_{\phi = 0} =
\left( (1 + \tau_R) \lambda_L^u - \frac{\delta_R}{\lambda_L^u} \right) \frac{\lambda_L^u}{\lambda_L^u - \lambda_L^s}.
\label{eq:phiZeroCurveProof11}
\end{equation}
By inspection $\frac{\partial \phi}{\partial \tau_L} \big|_{\phi = 0} < 0$.
Thus $G$ is unique (because if $\phi = 0$ for two distinct values of $\tau_L > \delta_L + 1$
then $\frac{\partial \phi}{\partial \tau_L} \ge 0$ at at least one of these values).

Since $\phi(\xi)$ is $C^\infty$ the function $G$ is $C^\infty$ by the implicit function theorem.
From \eqref{eq:phi2} we obtain
\begin{equation}
\frac{\partial \phi}{\partial \tau_R} =
\lambda_L^u \left( \lambda_L^u - 1 \right),
\label{eq:phiZeroCurveProof20}
\end{equation}
which is evidently positive.
Thus $\frac{d G}{d \tau_R} = -\frac{\frac{\partial \phi}{\partial \tau_L}}{\frac{\partial \phi}{\partial \tau_R}} \Big|_{\phi = 0} > 0$,
so $G$ is strictly increasing.

Also $G(\tau_R) \to \delta_L+1$ as $\tau_R \to -\infty$ because
if we fix $\tau_L = \delta_L + 1 + \ee$,
then $\phi(\xi) \to -\infty$ as $\tau_R \to -\infty$ for any $\ee > 0$.
Finally, by substituting $\tau_R = -\delta_R - 1$ into \eqref{eq:phi2} we obtain
\begin{equation}
\phi(\xi) \big|_{\tau_R = -\delta_R - 1} = -\delta_R \lambda_L^{u^2} + (1-\delta_L) \lambda_L^u + \delta_R \,.
\label{eq:phiZeroCurveProof30}
\end{equation}
Since $\tau_L = \lambda_L^u + \frac{\delta_L}{\lambda_L^u}$
we have $G(-\delta_R-1) = \alpha + \frac{\delta_L}{\alpha}$.
\end{proof}

\subsection{The curve $\hat{\psi}(\xi) = 0$}
\label{sub:psi}

The arguments presented here for $\hat{\psi}$ mirror those above for $\phi$.
We first show $\hat{\psi}(\xi) = 0$ does not exist
in $\Phi_{\rm slice}(\delta_L,\delta_R)$ if $\delta_R \ge 1$.

\begin{lemma}
Let $\xi \in \Phi$.
If $\delta_R \ge 1$ then $\hat{\psi}(\xi) < 0$.
\label{le:deltaRge1}
\end{lemma}

\begin{proof}
By inspection the first two terms in \eqref{eq:psiHat} are negative
and if $\delta_R \ge 1$ then the last term is less than or equal to zero.
\end{proof}

We now show $\hat{\psi}(\xi) = 0$ appears roughly as in Fig.~\ref{fig:igRN_phipsi}.

\begin{proposition}
Let $\delta_L > 0$ and $0 < \delta_R < 1$.
There exists a unique $C^\infty$ function $H : [\delta_L+1,\infty) \to (-\infty,-\delta_R-1)$
such that
\begin{equation}
\hat{\psi} \big( \tau_L, \delta_L, H(\tau_L), \delta_R \big) = 0,
\label{eq:psiHatZeroCurve}
\end{equation}
for all $\tau_L \in [\delta_L+1,\infty)$.
Moreover,
$H$ is strictly increasing,
$H(\tau_L) \to -\delta_R-1$ as $\tau_L \to \infty$,
and $H(\delta_L+1) = \beta + \frac{\delta_R}{\beta}$ where $\beta \in \mathbb{R}$ is the smallest (most negative) solution to $p(\beta) = 0$ where
\begin{equation}
p(\beta) = (1+\delta_L) \beta^3 + (1-\delta_L-\delta_R) \beta^2 - (1+\delta_L) \beta + \delta_L \,.
\label{eq:psiHatLeftIntersection}
\end{equation}
\label{pr:psiHatZeroCurve}
\end{proposition}

\begin{proof}
Fix $\tau_L \ge \delta_L + 1$.
With $\tau_R = -\delta_R - 1$ we have $\lambda_R^u = -1$ and so \eqref{eq:psiHat} simplifies to
$\hat{\psi}(\xi) = 1 - \delta_R > 0$.
Also $\hat{\psi}(\xi) \to -\infty$ as $\tau_R \to -\infty$,
thus, by the intermediate value theorem,
there exists $\tau_R = H(\tau_L) < -\delta_R-1$ satisfying \eqref{eq:psiHatZeroCurve}.

From \eqref{eq:psiHat},
\begin{equation}
\frac{\partial \hat{\psi}}{\partial \tau_R} =
\left( 3 \tau_L \lambda_R^{u^2} + 2 (1 - \delta_L - \delta_R) \lambda_R^u - \tau_L \right)
\frac{\lambda_R^u}{\lambda_R^u - \lambda_R^s},
\nonumber
\end{equation}
and if $\hat{\psi}(\xi) = 0$ this can be simplified to
\begin{equation}
\frac{\partial \hat{\psi}}{\partial \tau_R} \bigg|_{\hat{\psi}=0} =
\left( \tau_L \left( 1 + \lambda_R^{u^2} \right) - \frac{2 \delta_L}{\lambda_R^u} \right)
\frac{\lambda_R^u}{\lambda_R^u - \lambda_R^s},
\label{eq:psiHatZeroCurveProof11}
\end{equation}
which is positive.
Hence $H(\tau_L)$ satisfying \eqref{eq:psiHatZeroCurve} is unique for all $\tau_L \ge \delta_L + 1$.
Moreover, $H$ is $C^\infty$ because $\hat{\psi}$ is $C^\infty$.
From \eqref{eq:psiHat},
\begin{equation}
\frac{\partial \hat{\psi}}{\partial \tau_L} =
\lambda_R^u \left( \lambda_R^{u^2} - 1 \right) < 0,
\nonumber
\end{equation}
thus $\frac{d H}{d \tau_L} =
-\frac{\frac{\partial \hat{\psi}}{\partial \tau_R}}{\frac{\partial \hat{\psi}}{\partial \tau_L}}
\Big|_{\hat{\psi} = 0} > 0$,
i.e.~$H$ is strictly increasing.

We have $H(\tau_L) \to -\delta_R-1$ as $\tau_L \to \infty$ because if $\tau_R = -\delta_R-1-\ee$
then $\hat{\psi}(\xi) \to -\infty$ as $\tau_L \to \infty$ for any $\ee > 0$.
Finally, by substituting $\tau_L = \delta_L+1$ into \eqref{eq:psiHat} we obtain
$\hat{\psi}(\xi) \big|_{\tau_L = \delta_L + 1} = p \left( \lambda_R^u \right)$
and so $H(\delta_L+1) = \beta + \frac{\delta_R}{\beta}$ as required.
\end{proof}

Next we obtain upper bounds on the values of $\beta$ and $\beta + \frac{\delta_R}{\beta}$.
These are the values of $\lambda_R^u$ and $\tau_R$ for the point at which the curve $\hat{\psi}(\xi) = 0$
meets the boundary $\tau_L = \delta_L + 1$, see Fig.~\ref{fig:igRN_phipsi}.

\begin{lemma}
Let $\delta_L > 0$ and $0 < \delta_R < 1$.
The value of $\beta$ in Proposition \ref{pr:psiHatZeroCurve} satisfies $\beta > -\frac{1 + \sqrt{5}}{2}$
and $\beta + \frac{\delta_R}{\beta} > -2$.
\label{le:betaBounds}
\end{lemma}

\begin{proof}
The function $p$ can be rewritten as
\begin{equation}
p(\beta) =
\delta_L \left( \beta - 1 \right)^2 \left( \beta + 1 \right)
- \delta_R \beta^2
+ \beta \left( \beta^2 + \beta - 1 \right).
\label{eq:betaBoundsProof10}
\end{equation}
The first two terms of \eqref{eq:betaBoundsProof10} are negative,
so since $p(\beta) = 0$ the last term of \eqref{eq:betaBoundsProof10} must be positive.
This requires $\beta > -\frac{1 + \sqrt{5}}{2}$.

Also $p$ can be rewritten as
\begin{equation}
p(\beta) = \left[ (\beta+1) \left( 1 - \frac{1}{\beta} \right) \left( 1 + \delta_L - \frac{\delta_L}{\beta} \right) + (1-\delta_R) \right] \beta^2.
\nonumber
\end{equation}
Thus $p(\beta) = 0$ implies
\begin{equation}
-(\beta + 1) = \frac{1-\delta_R}{\left( 1 - \frac{1}{\beta} \right) \left( 1 + \delta_L - \frac{\delta_L}{\beta} \right)}.
\label{eq:betaBoundsProof30}
\end{equation}
Since $\beta < 0$ the denominator of \eqref{eq:betaBoundsProof30} is greater than $1$
and so $-(\beta + 1) < 1-\delta_R$.
Thus $(\beta+1)^2 < (1-\delta_R)^2$ which can be rearranged as
$\beta^2 + \delta_R < -2 \beta - \delta_R (1 - \delta_R)$.
Since $0 < \delta_R < 1$ this can be reduced to $\beta + \frac{\delta_R}{\beta} > -2$.
\end{proof}

Lastly we show that the curve
$\tau_R = -\frac{1}{\tau_L} - \delta_R - 1$
lies below $\hat{\psi}(\xi) = 0$, as in Fig.~\ref{fig:igRN_phipsi}.
This result is used later in the proof of Proposition \ref{pr:psiHat}.

\begin{lemma}
Let $\delta_L > 0$, $0 < \delta_R < 1$, and $\tau_L \ge \delta_L + 1$.
Then
\begin{equation}
H(\tau_L) > -\frac{1}{\tau_L} - \delta_R - 1.
\label{eq:HBound}
\end{equation}
\label{le:HBound}
\end{lemma}

\begin{proof}
By iterating \eqref{eq:T} under $f_{R,\xi}$ and $f_{L,\xi}$ we obtain
\begin{equation}
f_\xi^2(T) = \left(
\tau_L \left( \frac{\tau_R}{1 - \lambda_R^s} + 1 \right) - \frac{\delta_R}{1 - \lambda_R^s} + 1,
-\delta_L \left( \frac{\tau_R}{1 - \lambda_R^s} + 1 \right) \right).
\label{eq:f2T}
\end{equation}
The second component of \eqref{eq:f2T} is clearly positive with any $\tau_R < -\delta_R - 1$.
The first component of \eqref{eq:f2T} can be rearranged as
\begin{equation}
f_\xi^2(T)_1 =
\left( \tau_L - \frac{(\tau_R + \delta_R) \lambda_R^s - 1}{\tau_R + \delta_R + 1} \right)
\left( \frac{\tau_R}{1 - \lambda_R^s} + 1 \right).
\label{eq:f2T1}
\end{equation}
If $\tau_L = \frac{-1}{\tau_R + \delta_R + 1}$
(equivalently $\tau_R = -\frac{1}{\tau_L} - \delta_R - 1$)
then \eqref{eq:f2T1} simplifies to a quantity that is clearly negative.
In this case $f_\xi^2(T)$ is located in the second quadrant of $\mathbb{R}^2$,
so certainly it lies to the left of $E^s(X)$.
Thus $\psi(\xi) > 0$ by Proposition \ref{pr:psi}, so $\hat{\psi}(\xi) < 0$.

We have shown $\tau_R = -\frac{1}{\tau_L} - \delta_R - 1$ implies $\hat{\psi}(\xi) < 0$.
Therefore if $\hat{\psi}(\xi) = 0$ (equivalently $\tau_R = H(\tau_L)$),
then $\tau_R > -\frac{1}{\tau_L} - \delta_R - 1$, as required.
\end{proof}

\subsection{The curves $\phi(\xi) = 0$ and $\hat{\psi}(\xi) = 0$ intersect at a unique point}
\label{sub:phiandpsi}

\begin{proposition}
Fix $0 < \delta_L < 1$ and $0 < \delta_R < 1$.
There exist unique $\tau_L > \delta_L + 1$ and $\tau_R < -\delta_R-1$
such that $\phi(\xi) = \hat{\psi}(\xi) = 0$.
\label{pr:phipsiIntersection}
\end{proposition}

\begin{proof}
By Propositions \ref{pr:phiZeroCurve} and \ref{pr:psiHatZeroCurve}
the curves $\phi(\xi) = 0$ and $\hat{\psi}(\xi) = 0$ must intersect.
To show this intersection is unique it suffices to show that at any point of intersection
the slope $\frac{d \tau_R}{d \tau_L}$ of $\phi(\xi) = 0$ is greater than that of $\hat{\psi}(\xi) = 0$.

From the calculations performed in the proof of Proposition \ref{pr:phiZeroCurve},
the slope of $\phi(\xi) = 0$ is
\begin{equation}
\left( \frac{d G}{d \tau_R} \right)^{-1} =
\frac{-(1 + \tau_R) \lambda_L^u + \frac{\delta_R}{\lambda_L^u}}
{\left( \lambda_L^u - 1 \right) \left( \lambda_L^u - \lambda_L^s \right)}.
\nonumber
\end{equation}
Consequently
\begin{equation}
\left( \frac{d G}{d \tau_R} \right)^{-1} > -\frac{\lambda_R^u + 1}{\lambda_L^u - 1},
\label{eq:phiZeroSlopeApprox}
\end{equation}
because $\tau_R < \lambda_R^u$, $\delta_R > 0$, and $\lambda_L^s > 0$.
From the calculations performed in the proof of Proposition \ref{pr:psiHatZeroCurve},
the slope of $\hat{\psi}(\xi) = 0$ is
\begin{equation}
\frac{d H}{d \tau_L} =
\frac{\left( \lambda_R^{u^2} - 1 \right) \left( \lambda_R^u - \lambda_R^s \right)}
{\tau_L \left( 1 + \lambda_R^{u^2} \right) - \frac{2 \delta_L}{\lambda_R^u}}.
\nonumber
\end{equation}
Consequently
\begin{equation}
\frac{d H}{d \tau_L} < -\frac{\lambda_R^u \left( \lambda_R^{u^2} - 1 \right)}{\lambda_L^u \left( \lambda_R^{u^2} + 1 \right)},
\label{eq:psiHatZeroSlopeApprox}
\end{equation}
because $\tau_L > \lambda_L^u$, $\delta_L > 0$, and $\lambda_R^s < 0$.

Now suppose for a contradiction that $\left( \frac{d G}{d \tau_R} \right)^{-1} \le \frac{d H}{d \tau_L}$
at a point where both $\phi(\xi) = 0$ and $\hat{\psi}(\xi) = 0$.
By \eqref{eq:phiZeroSlopeApprox} and \eqref{eq:psiHatZeroSlopeApprox} this implies
\begin{equation}
-\frac{\lambda_R^u + 1}{\lambda_L^u - 1} < -\frac{\lambda_R^u \left( \lambda_R^{u^2} - 1 \right)}{\lambda_L^u \left( \lambda_R^{u^2} + 1 \right)},
\nonumber
\end{equation}
which can be rearranged as
\begin{equation}
-\frac{\left( \lambda_R^u + 1 \right)
\left[ \lambda_L^u \left( \lambda_R^u + 1 \right) + \lambda_R^u \left( \lambda_R^u - 1 \right) \right]}
{\lambda_L^u \left( \lambda_L^u - 1 \right) \left( \lambda_R^{u^2} + 1 \right)} < 0.
\nonumber
\end{equation}
For this to be true the term in square brackets must be negative, and this implies
\begin{equation}
\lambda_L^u (\tau_R + 1) < -2,
\label{eq:phipsiIntersectionProof50}
\end{equation}
because $\tau_R < \lambda_R^u$ and $\lambda_R^u \left( \lambda_R^u - 1 \right) > 2$.
However, $\phi(\xi) = 0$, so by applying the quadratic formula to \eqref{eq:phi2} we obtain
\begin{equation}
\tau_R + \delta_L + \delta_R - \sqrt{(\tau_R + \delta_L + \delta_R)^2 - 4 (1 + \tau_R) \delta_R} = 2 \lambda_L^u (\tau_R + 1).
\nonumber
\end{equation}
Thus \eqref{eq:phipsiIntersectionProof50} implies
\begin{equation}
\tau_R + \delta_L + \delta_R - \sqrt{(\tau_R + \delta_L + \delta_R)^2 - 4 (1 + \tau_R) \delta_R} < -4,
\nonumber
\end{equation}
which can be rearranged as
\begin{equation}
\tau_R < \frac{-2 \delta_L - 3 \delta_R - 4}{2 + \delta_R}.
\nonumber
\end{equation}
Since $\delta_L, \delta_R > 0$ this implies $\tau_R < -2$.
But the curve $\hat{\psi}(\xi) = 0$ increases with $\tau_L$,
thus on $\hat{\psi}(\xi) = 0$ the value of $\tau_R$ is greater than its
value at the boundary $\tau_L = \delta_L + 1$
where it equals $\beta + \frac{\delta_R}{\beta}$.
So the bound $\beta + \frac{\delta_R}{\beta} > -2$ of Lemma \ref{le:betaBounds} provides a contradiction.
Therefore $\left( \frac{d G}{d \tau_R} \right)^{-1} > \frac{d H}{d \tau_L}$
at any point where $\phi(\xi) = 0$ and $\hat{\psi}(\xi) = 0$ intersect,
hence the intersection point is unique.
\end{proof}

\section{Dynamics of the renormalisation operator}
\label{sec:renormalisation}
\setcounter{equation}{0}

In this section we study the dynamics of $g$ on $\Phi$.
We first show that any $\xi \in \Phi$ maps under $g$ to another point in $\Phi$.

\begin{proposition}
If $\xi \in \Phi$ then $g(\xi) \in \Phi$.
\label{pr:PhiForwardInvariant}
\end{proposition}

\begin{proof}
Write $g(\xi) = \left( \tilde{\tau}_L, \tilde{\delta}_L, \tilde{\tau}_R, \tilde{\delta}_R \right)$.
By \eqref{eq:g} and the assumption $\xi \in \Phi$ we obtain
\begin{align*}
\tilde{\tau}_L - \left( \tilde{\delta}_L + 1 \right)
&= \tau_R^2 - 2 \delta_R - \left( \delta_R^2 + 1 \right)
= \tau_R^2 - \left( \delta_R + 1 \right)^2 > 0, \\
\tilde{\delta}_L &= \delta_R^2 > 0, \\
\tilde{\tau}_R + \tilde{\delta}_R + 1
&= \tau_L \tau_R - \delta_L - \delta_R + \delta_L \delta_R + 1 \\
&< -(\delta_L + 1)(\delta_R + 1) - \delta_L - \delta_R + \delta_L \delta_R + 1 \\
&= -2 (\delta_L + \delta_R) < 0, \\
\tilde{\delta}_R &= \delta_L \delta_R > 0,
\end{align*} 
which implies $g(\xi) \in \Phi$.
\end{proof}

Next in \S\ref{sub:renormalisation} we consider the subset of $\Phi$ for which $\hat{\psi}(\xi) < 0$.
We show that any point in this subset
maps under $g$ to another point in this subset.
This result is central to showing that the regions $\cR_n$ are mutually disjoint
and proving Theorem \ref{th:Rn} in \S\ref{sub:RnProof}.
Recall, the sign of $\hat{\psi}(\xi)$ is the same as that of $\zeta_1(\xi)$ by \eqref{eq:psiHatDefn}.

\subsection{The subset of $\Phi$ for which $\hat{\psi}(\xi) < 0$}
\label{sub:renormalisation}

We first show that the point at which the curve $\hat{\psi}(\xi) = 0$
meets $\tau_L = \delta_L + 1$
maps under $g$ to a point below the dashed curve of Fig.~\ref{fig:igRN_phipsi}
in the corresponding slice $\Phi_{\rm slice}(\tilde{\delta}_L,\tilde{\delta}_R)$.

\begin{lemma}
Let $\delta_L > 0$ and $0 < \delta_R < 1$.
Let $\xi_0 = (\delta_L+1,\delta_L,\beta + \frac{\delta_R}{\beta},\delta_R)$ 
where $\beta$ is as given in Proposition \ref{pr:psiHatZeroCurve}.
Write $g(\xi_0) = \left( \tilde{\tau}_L, \tilde{\delta}_L, \tilde{\tau}_R, \tilde{\delta}_R \right)$.
Then
\begin{equation}
\tilde{\tau}_R < -\frac{1}{\tilde{\tau}_L} - \tilde{\delta}_R - 1.
\label{eq:tildeTauRBound}
\end{equation}
\label{le:tildeTauRBound}
\end{lemma}

\begin{proof}
The inequality \eqref{eq:tildeTauRBound} is equivalent to
\begin{equation}
\tilde{\tau_L} \left( \tilde{\tau}_R + \tilde{\delta}_R + 1 \right) + 1 < 0.
\label{eq:tildeTauRBoundProof1}
\end{equation}
By \eqref{eq:g} we have
$\tilde{\tau}_L = \tau_R^2 - 2 \delta_R$,
$\tilde{\tau}_R = \tau_L \tau_R - \delta_L - \delta_R$, and
$\tilde{\delta}_R = \delta_L \delta_R$;
also $\tau_L = \delta_L + 1$.
Upon substituting these into \eqref{eq:tildeTauRBoundProof1},
after simplification the left-hand side of \eqref{eq:tildeTauRBoundProof1} becomes
\begin{equation}
\omega = (1+\delta_L) \tau_R^3 + (1-\delta_L)(1-\delta_R) \tau_R^2 - 2 \delta_R (1+\delta_L) \tau_R
- 2 \delta_R (1-\delta_L)(1-\delta_R) + 1.
\label{eq:omega}
\end{equation}
Thus it remains for us to show that $\omega < 0$.

Into \eqref{eq:omega} we substitute $\tau_R = \beta + \frac{\delta_R}{\beta}$ to obtain,
after much rearranging,
\begin{equation}
\omega = p(\beta) + q(\beta) + \delta_L \delta_R \beta (\beta + 2)
+ (1-\delta_L) (\beta+1)
+ \delta_R^2 (1+\delta_L) \left( \beta + \frac{\delta_R}{\beta} \right) \frac{1}{\beta^2},
\label{eq:tildeTauRBoundProof10}
\end{equation}
where $p$ is given by \eqref{eq:psiHatLeftIntersection} and
\begin{equation}
q(\beta) = \big( \delta_L (2-\delta_R) + \delta_R \big) \beta
+ \delta_R^2 (1-\delta_L)(1-\delta_R) \frac{1}{\beta^2}.
\label{eq:tildeTauRBoundProof11}
\end{equation}
Since $\beta < -1$ we have
\begin{align}
q(\beta) &< -\big( \delta_L (2-\delta_R) + \delta_R \big) + \delta_R^2 (1-\delta_L)(1-\delta_R) \nonumber \\
&< -\big( \delta_L (2-\delta_R) + \delta_R \big) + \delta_R^2 (1-\delta_R) \nonumber \\
&= -\delta_L (2 - \delta_R) - \delta_R \left( \delta_R^2 - \delta_R + 1 \right) \nonumber \\
&< 0. \nonumber
\end{align}
Also $p(\beta) = 0$
and by inspection the last three terms of \eqref{eq:tildeTauRBoundProof10} are negative
(because $\beta+1 < 0$ and $\beta + 2 > 0$ by Lemma \ref{le:betaBounds}).
Therefore $\omega < 0$.
\end{proof}

We now use Lemma \ref{le:tildeTauRBound} to show that
the subset of $\Phi$ for which $\hat{\psi}(\xi) < 0$ is forward invariant under $g$.

\begin{proposition}
Let $\xi \in \Phi$.
If $\hat{\psi}(\xi) \le 0$ then $\hat{\psi}(g(\xi)) < 0$.
\label{pr:psiHat}
\end{proposition}

\begin{proof}[Proof of Proposition \ref{pr:psiHat}]
Write $g(\xi) = \left( \tilde{\tau}_L, \tilde{\delta}_L, \tilde{\tau}_R, \tilde{\delta}_R \right)$.
Since $\xi \in \Phi$ we have $\delta_L, \delta_R > 0$.

First suppose $0 < \delta_R < 1$.
If $\tilde{\delta}_R \ge 1$ then certainly $\hat{\psi}(g(\xi)) < 0$ by Lemma \ref{le:deltaRge1},
so let us suppose $\tilde{\delta}_R < 1$.
Since $\tilde{\delta}_L = \delta_R^2 < 1$,
by Proposition \ref{pr:phipsiIntersection} the curves
$\phi = 0$ and $\hat{\psi} = 0$ intersect at a unique point in $\Phi_{\rm slice}(\tilde{\delta}_L,\tilde{\delta}_R)$,
call it $\tilde{\xi}_{\rm int}$, see Fig.~\ref{fig:igRN_phipsiImage}.
With $\xi = \xi_0$ as in Lemma \ref{le:tildeTauRBound},
the inequality \eqref{eq:tildeTauRBound} implies $\hat{\psi}(g(\xi_0)) < 0$ by Lemma \ref{le:HBound}.
Also $\phi(g(\xi_0)) = 0$, because $\hat{\psi}(\xi_0) = 0$, thus
$g(\xi_0)$ lies on $\phi = 0$ and below $\tilde{\xi}_{\rm int}$, as in Fig.~\ref{fig:igRN_phipsiImage}.

Now if $\hat{\psi}(\xi) \le 0$ and $\xi \ne \xi_0$,
then $g(\xi)$ lies in the shaded region of Fig.~\ref{fig:igRN_phipsiImage}.
The curve $\hat{\psi} = 0$ does not enter this region because the intersection point $\tilde{\xi}_{\rm int}$ is unique.
Thus $g(\xi)$ lies below the curve $\hat{\psi} = 0$, that is $\hat{\psi}(g(\xi)) < 0$.

Second suppose $\delta_R \ge 1$.
Then
\begin{equation}
\tilde{\tau}_R
= \tau_L \tau_R - \delta_L - \delta_R
< -(\delta_L+1)(\delta_R+1) - \delta_L - \delta_R
< -3,
\nonumber
\end{equation}
where we have used $\delta_L > 0$ and $\delta_R \ge 1$ to produce the last inequality.
Thus $\tilde{\tau}_R < -2$ and so $g(\xi)$ lies below $\hat{\psi} = 0$ by Lemma \ref{le:betaBounds}.
That is, $\hat{\psi}(g(\xi)) < 0$.
\end{proof}

\begin{figure}[t!]
\begin{center}
\includegraphics[height=10cm]{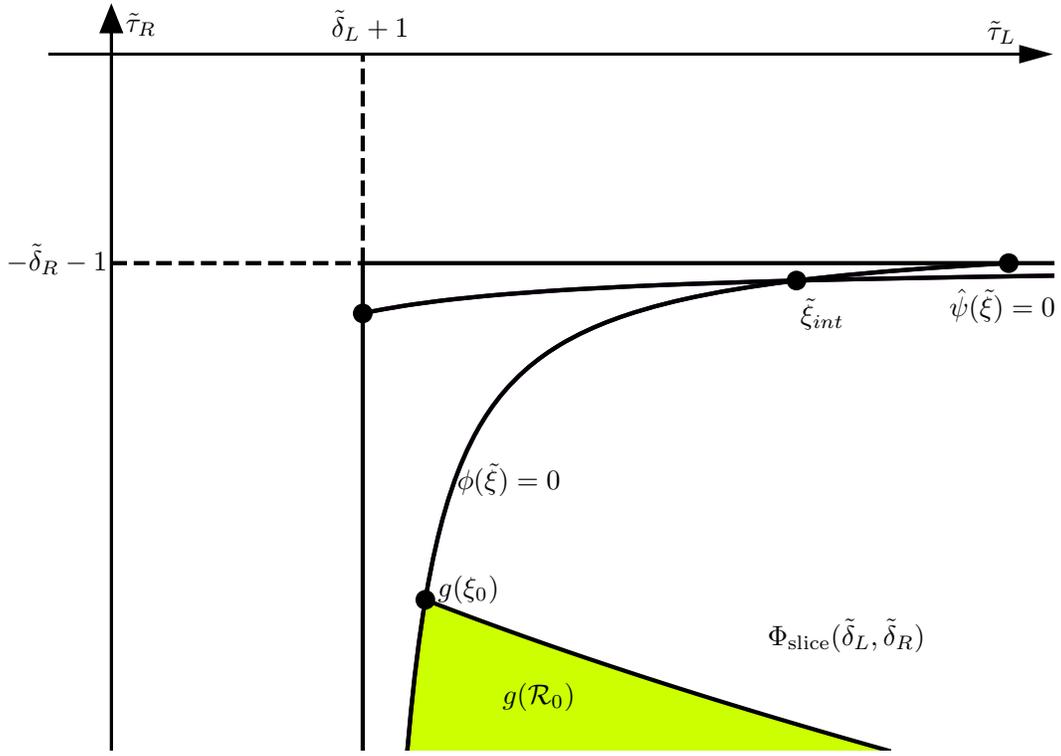}
\caption{
A sketch of $\phi(\tilde{\xi}) = 0$ and $\hat{\psi}(\tilde{\xi}) = 0$
where $\tilde{\xi} = g(\xi)$ with $0 < \tilde{\delta}_L < 1$ and $0 < \tilde{\delta}_R < 1$.
The point $\tilde{\xi}_{\rm int}$ is the unique intersection of $\phi(\tilde{\xi}) = 0$ and $\hat{\psi}(\tilde{\xi}) = 0$.
The point $\xi_0$ is as in Lemma \ref{le:tildeTauRBound}.
\label{fig:igRN_phipsiImage}
} 
\end{center}
\end{figure}

\subsection{Arguments leading to a proof of Theorem \ref{th:Rn}}
\label{sub:RnProof}

Here we prove Theorem \ref{th:Rn} after a sequence of lemmas.

\begin{lemma}
Let $\xi \in \cR_n$ for some $n \ge 1$.
Then $g^i(\xi) \in \cR_{n-i}$ for all $i = 1,2,\ldots,n$.
\label{le:gForwards}
\end{lemma}

\begin{proof}
We have $\zeta_n(\xi) > 0$ and $\zeta_{n+1}(\xi) \le 0$ by \eqref{eq:Rn}.
Thus $\zeta_{n-i} \left( g^i(\xi) \right) > 0$ and $\zeta_{n-i+1} \left( g^i(\xi) \right) \le 0$ by \eqref{eq:zetan}.
Also $g^i(\xi) \in \Phi$ by Proposition \ref{pr:PhiForwardInvariant}.
Thus $g^i(\xi) \in \cR_{n-i}$ by \eqref{eq:Rn}.
\end{proof}

\begin{lemma}
Let $\xi \in \Phi$ with $g(\xi) \in \cR_{n-1}$ for some $n \ge 1$.
Then $\xi \in \cR_n$.
\label{le:gBackwards}
\end{lemma}

\begin{proof}
We have $\zeta_{n-1}(g(\xi)) > 0$ and $\zeta_n(g(\xi)) \le 0$ by \eqref{eq:Rn}.
Thus $\zeta_n(\xi) > 0$ and $\zeta_{n+1}(\xi) \le 0$ by \eqref{eq:zetan}.
So $\xi \in \cR_n$ because also $\xi \in \Phi$.
\end{proof}

\begin{lemma}
Let $\xi \in \cR_n$ for some $n \ge 1$.
Then $\zeta_0(g(\xi)) > 0$.
\label{le:zeta0gxi}
\end{lemma}

\begin{proof}
We have $\zeta_n(\xi) > 0$ by \eqref{eq:Rn},
thus $\zeta_1 \left( g^{n-1}(\xi) \right) > 0$ by \eqref{eq:zetan}.
Thus $\zeta_1(\xi) > 0$ by Proposition \ref{pr:psiHat}
(recall the sign of $\zeta_1$ is the same as that of $\hat{\psi}$).
That is, $\zeta_0(g(\xi)) > 0$.
\end{proof}

\begin{lemma}
Let $\xi \in \Phi$ and write
$g^i(\xi) = \left( \tau_{L,i}, \delta_{L,i}, \tau_{R,i}, \delta_{R,i} \right)$ for each $i$.
Then $\tau_{L,2} > \tau_L^2 \tau_R^2$ and $\tau_{R,2} < \tau_L \tau_R$.
\label{eq:g2Bound}
\end{lemma}

\begin{proof}
By \eqref{eq:g},
\begin{equation}
\tau_{L,2} = \tau_{R,1}^2 - 2 \delta_{R,1} = \left( \tau_L \tau_R - \delta_L - \delta_R \right)^2 - 2 \delta_L \delta_R \,, \nonumber
\end{equation}
which can be rearranged as
\begin{equation}
\tau_{L,2} = \left( \tau_L \tau_R - \delta_L \right)^2
+ \left( \tau_L \tau_R - \delta_R \right)^2
- \tau_L^2 \tau_R^2 \,.
\nonumber
\end{equation}
Then from the bounds in \eqref{eq:saddleSaddleRegion} we obtain $\tau_{L,2} > \tau_L^2 \tau_R^2$.
Also
\begin{equation}
\tau_{R,2} = \tau_{L,1} \tau_{R,1} - \delta_{L,1} - \delta_{R,1} < \tau_{L,1} \tau_{R,1} \,.
\nonumber
\end{equation}
By substituting $\tau_{L,1} > 1$ and $\tau_{R,1} > \tau_L \tau_R$ we obtain
$\tau_{R,2} < \tau_L \tau_R$.
\end{proof}

\begin{proof}[Proof of Theorem \ref{th:Rn}]
Suppose for a contradiction that the $\cR_n$ are {\em not} mutually disjoint.
So there exists $\xi \in \cR_m \cap \cR_n$ for some $0 \le m < n$.
This implies $g^{n-1}(\xi) \in \cR_1$ by Lemma \ref{le:gForwards},
and so $\hat{\psi}(g^{n-1}(\xi)) > 0$ (the sign of $\zeta_1$ is the same as that of $\hat{\psi}$).
Also $g^m(\xi) \in \cR_0$, so $\hat{\psi}(g^m(\xi)) \le 0$.
By Proposition \ref{pr:psiHat}, $\hat{\psi}(g^{m+i}(\xi)) \le 0$ for all $i \ge 0$.
In particular $\hat{\psi}(g^{n-1}(\xi)) \le 0$, and this is a contradiction.
Therefore the $\cR_n$ are mutually disjoint.

Now choose any $\xi \in \Phi_{\rm BYG}$.
To verify \eqref{eq:RnUnion} we show there exists $n \ge 0$ such that $\xi \in \cR_n$.
Certainly this is true if $\hat{\psi}(\xi) \le 0$, because in this case $\xi \in \cR_0$,
so let us assume $\hat{\psi}(\xi) > 0$.
In view of Lemma \ref{eq:g2Bound}, we consider the map $\tilde{g} : \mathbb{R}^2 \to \mathbb{R}^2$ defined by
\begin{equation}
\tilde{g}(\tau_L,\tau_R) = \left( \left( \tau_L \tau_R \right)^2, \tau_L \tau_R \right).
\nonumber
\end{equation}
For any $j \ge 1$ the $j^{\rm th}$-iterate of $\tilde{g}$ is given explicitly by
\begin{equation}
\tilde{g}^j(\tau_L,\tau_R) = \left( \left( \tau_L \tau_R \right)^{2 k_j}, \left( \tau_L \tau_R \right)^{k_j} \right),
\nonumber
\end{equation}
where $k_j = 3^{j-1}$.
Then Lemma \ref{eq:g2Bound} implies
$\tau_{R,2 j} < \left( \tau_L \tau_R \right)^{k_j}$ (using the notation of Lemma \ref{eq:g2Bound})
and so $\tau_{R,2 j} \to -\infty$ as $j \to \infty$.
Thus there exists $m \ge 0$ such that $\tau_{R,m} \le -2$.
Then $\hat{\psi}(g^m(\xi)) < 0$ by Lemma \ref{le:betaBounds}.
Now let $n \in \{ 1,2,\ldots, m \}$ be the smallest integer for which $\hat{\psi}(g^n(\xi)) \le 0$.
Then $\hat{\psi}(g^{n-1}(\xi)) > 0$, so $\phi(g^n(\xi)) > 0$.
That is, $g^n(\xi) \in \cR_0$.
Hence $\xi \in \cR_n$, by $n$ applications of Lemma \ref{le:gBackwards}.
This completes our verification of \eqref{eq:RnUnion}.

To show that $\cR_j$ is non-empty for all $j \ge 0$, first observe $\hat{\psi}(\xi^*) > 0$.
Also $\cR_0$ is certainly non-empty.
So for any $j \ge 1$ we can choose $\xi \in \Phi_{\rm BYG}$ sufficiently close to $\xi^*$
that $\hat{\psi}(g^i(\xi)) > 0$ for all $i = 0,1,\ldots,j-1$.
Again let $n \ge 1$ be the smallest integer for which $\hat{\psi}(g^n(\xi)) \le 0$.
Then $n \ge j$ and $g^n(\xi) \in \cR_0$.
Thus $g^{n-j}(\xi) \in \cR_j$ (by again using Lemma \ref{le:gBackwards}), i.e.~$\cR_j$ is non-empty.

Finally, choose any $\ee > 0$ and
let $B_\ee(\xi^*)$ be the open ball in $\mathbb{R}^4$ centred at $\xi^*$ and with radius $\ee$ using the Euclidean norm.
We now show there exists $m \ge 1$ such that $\cR_n \subset B_\ee(\xi^*)$ for all $n > m$.
This will prove that $\cR_n \to \{ \xi^* \}$ as $n \to \infty$.
Choose any $\xi \in \Phi$ with $\xi \notin B_\ee(\xi^*)$.
It is simple exercise to show that $|\tau_L \tau_R| \ge 1 + \frac{\ee}{\sqrt{2}}$.
Thus, as above, there exists $m \ge 0$ such that $\tau_{R,m} \le -2$
and $\xi \in \cR_n$ for some $n \le m$.
Hence for any $n > m$ the region $\cR_n$ contains no points outside of $B_\ee(\xi^*)$.
That is $\cR_n \subset B_\ee(\xi^*)$ for all $n > m$ and therefore $\cR_n \to \{ \xi^* \}$ as $n \to \infty$.
\end{proof}

\section{Positive Lyapunov exponents}
\label{sec:lyap}
\setcounter{equation}{0}


For smooth maps Lyapunov exponents are usually defined in terms of the derivative of the map.
The border-collision normal form $f_\xi$ is not differentiable on $x=0$,
so instead we work with one-sided directional derivatives, \S\ref{sub:osdd}.
We then define Lyapunov exponents in terms of these derivatives, \S\ref{sub:lyap}.
This definition coincides with the familiar interpretation of Lyapunov exponents as the
asymptotic rate of separation of nearby forward orbits \cite{Si20e}.
Then in \S\ref{sub:R0Proof} we prove Theorem \ref{th:R0}.

\subsection{One-sided directional derivatives}
\label{sub:osdd}

\begin{definition}
The {\em one-sided directional derivative} of a function $F : \mathbb{R}^2 \to \mathbb{R}^2$
at $z \in \mathbb{R}^2$ in a direction $v \in \mathbb{R}^2$ is
\begin{equation}
\rD_v^+ F(z) = \lim_{\delta \to 0^+} \frac{F(z + \delta v) - F(z)}{\delta},
\label{eq:Dplus}
\end{equation}
if this limit exists.
\end{definition}

The following result tells us that one-sided directional derivatives
of the $n^{\rm th}$ iterate of \eqref{eq:f} exist everywhere and for all $n \ge 1$.
This follows from the piecewise-linearity and continuity of \eqref{eq:f}.
For a proof see \cite{Si20e}.

\begin{lemma}
For any $\xi \in \mathbb{R}^4$, $z \in \mathbb{R}^2$, $v \in \mathbb{R}^2$, and $n \ge 1$,
$\rD_v^+ f_\xi^n(z)$ exists.
\label{le:DplusExists}
\end{lemma}

\subsection{Lyapunov exponents}
\label{sub:lyap}

In view of Lemma \ref{le:DplusExists} we can use the following definition.

\begin{definition}
The {\em Lyapunov exponent} of $f_\xi$ at $z \in \mathbb{R}^2$ in a direction $v \in \mathbb{R}^2$ is
\begin{equation}
\lambda(z,v) = \limsup_{n \to \infty} \frac{1}{n} \ln \left( \left\| \rD_v^+ f_\xi^n(z) \right\| \right).
\label{eq:limsup}
\end{equation}
\end{definition}

If the forward orbit of $z$ does not intersect $x=0$,
then $\rD f_\xi^n(z)$ (the Jacobian matrix of $f_\xi^n$ at $z$) is well-defined for all $n \ge 1$.
Moreover, $\rD_v^+ f_\xi^n(z) = \rD f_\xi^n(z) v$, so in this case \eqref{eq:limsup} reduces to the usual
expression given for smooth maps.

The following result is Theorem 2.1 of \cite{GlSi21},
except in \cite{GlSi21} only forward orbits that do not intersect $x=0$ were considered.
The generalisation to one-sided directional derivatives is elementary so we do not provide a proof.
The proof in \cite{GlSi21} is achieved by constructing an invariant expanding cone
for multiplying vectors $v$ under the matrices $A_L$ and $A_R$.
The derivative in \eqref{eq:limsup} can be written as $v$ left-multiplied by $n$ matrices
each of which is either $A_L$ or $A_R$.
The cone implies the vector increases in norm each time it is multiplied by $A_L$ or $A_R$,
so certainly the norm increases on average, i.e.~$\lambda(z,v) > 0$.

\begin{proposition}
For any $\xi \in \Phi_{\rm BYG}$, $z \in \mathbb{R}^2$, and $v = (1,0)$,
\begin{equation}
\liminf_{n \to \infty} \frac{1}{n} \ln \left( \left\| \rD_v^+ f_\xi^n(z) \right\| \right) > 0.
\label{eq:liminf}
\end{equation}
\label{pr:lyap}
\end{proposition}


\subsection{Arguments leading to a proof of Theorem \ref{th:R0}}
\label{sub:R0Proof}

We are now ready to prove Theorem \ref{th:R0}.
Once we have constructed the set $\Delta$,
the equality \eqref{eq:LambdaAsInfiniteIntersection} follows from
the arguments given in the proof of Lemma 6.2 of \cite{GlSi21}.
We reproduce these arguments here for convenience.

\begin{proof}[Proof of Theorem \ref{th:R0}]
The set $\Lambda(\xi)$ is bounded because $X \in \Omega$ and $\Omega$ is bounded and forward invariant (Proposition \ref{pr:Omega}).
Also $\Lambda(\xi)$ is connected and invariant by the definition of an unstable manifold.
With $v = (1,0)$ and any $z \in \Lambda(\xi)$,
the Lyapunov exponent $\lambda(z,v)$ is well-defined by Lemma \ref{le:DplusExists}.
Moreover $\lambda(z,v) > 0$ by Proposition \ref{pr:lyap}
and because the supremum limit is greater than or equal to the infimum limit.

It remains for us to prove part (iii).
Here we assume $\delta_R < 1$; also $\delta_L < 1$ by Lemma \ref{le:deltaLge1}.
Since $\xi \in \cR_0$
we have $\zeta_1(\xi) \le 0$ and so $\psi(\xi) \ge 0$ by \eqref{eq:psi2}.
Thus $f_\xi^2(T)$ lies on or to the left of $E^s(X)$ by Proposition \ref{pr:psi}.
Let $Z$ denote the intersection of $E^s(X)$ with $\roverline{T f_\xi^2(T)}$
(the line segment connecting $T$ and $f_\xi^2(T)$).
Notice $\roverline{X T}$ and $\roverline{T Z}$ are subsets of $W^u(X)$
while $\roverline{Z X}$ is a subset of $W^s(X)$.

Let $\Delta_0$ be the filled triangle with vertices $X$, $T$, and $Z$, see Fig.~\ref{fig:igRN_X}-a.
Also let $\Delta = \bigcup_{n=0}^\infty f_\xi^n(\Delta_0)$.
The set $\Delta$ is forward invariant, by definition,
and has non-empty interior because it contains $\Delta_0$.
As in \cite{GlSi21}, let $\tilde{\Delta} = \bigcap_{n=0}^\infty f_\xi^n(\Delta)$.

We now show $\Lambda(\xi) \subset \tilde{\Delta}$.
Choose any $z \in \Lambda(\xi)$.
Let $\{ z_k \}$ be a sequence of points in $W^u(X)$ with $z_k \to z$ as $k \to \infty$.
For each $k$, $f_\xi^{-n}(z_k) \to X$ as $n \to \infty$,
thus there exists $n_k \ge 1$ such that $f_\xi^{-{n_k}}(z_k) \in \roverline{X T}$.
Thus $f_\xi^{-{n_k}}(z_k) \in \Delta_0$, so $z_k \in \Delta$.
This is true for all $k$, thus $z \in \Delta$.
But $z \in \Lambda(\xi)$ is arbitrary, thus $\Lambda(\xi) \subset \Delta$.
Also $\Lambda(\xi)$ is forward invariant, thus $\Lambda(\xi) \subset \tilde{\Delta}$.

Finally we show $\tilde{\Delta} \subset \Lambda(\xi)$.
The determinants $\delta_L$ and $\delta_R$ of the pieces of $f_\xi$ are both less than $1$,
thus the area (Lebesgue measure) of $f_\xi^n(\Delta)$ converges to $0$ as $n \to \infty$.
Now choose any $z \in \tilde{\Delta}$.
Then $z \in f_\xi^n(\Delta)$ for all $n \ge 0$ and so the distance of $z$ to the boundary of $f_\xi^n(\Delta)$ converges to $0$ as $n \to \infty$.
The boundary of $\Delta_0$ consists of $\roverline{X Z}$, which lies in the part of $W^s(X)$ that converges linearly to $X$,
and two line segments in $W^u(X)$.
Consequently the boundary of $f_\xi^n(\Delta_0)$ is contained in $\roverline{X f_\xi^n(Z)} \cup W^u(X)$ for all $n \ge 0$.
Thus the boundary of $\Delta$ is contained in $\roverline{Z f_\xi(Z)} \cup W^u(X)$,
so the boundary of $f_\xi^n(\Delta)$ is contained in $\roverline{f_\xi^n(Z) f_\xi^{n+1}(Z)} \cup W^u(X)$ for all $n \ge 0$.
But $\roverline{f_\xi^n(Z) f_\xi^{n+1}(Z)}$ converges to $X$, hence the distance of $z$ to $W^u(X)$ must be $0$.
Thus $z \in \Lambda(\xi)$.
But $z \in \tilde{\Delta}$ is arbitrary, thus $\tilde{\Delta} \subset \Lambda(\xi)$.
This completes our demonstration of \eqref{eq:LambdaAsInfiniteIntersection}.
\end{proof}

\section{Implementing the renormalisation recursively}
\label{sec:mainProof}
\setcounter{equation}{0}

In this section we work towards a proof of Theorem \ref{th:affinelyConjugate}.
First in \S\ref{sub:OmegaPrime} we use the unstable manifold of $X$ to construct a triangle $\Omega'(\xi)$
that maps to $\Omega(g(\xi))$ under the affine transformation $h_\xi$
for converting $f_\xi^2$ to $f_{g(\xi)}$.
In particular we show that $\Omega'(\xi)$ is a subset of both $\Omega(\xi)$ and $\Pi_\xi$
and this allows us to implement the renormalisation recursively in \S\ref{sub:proofByInduction}.

\subsection{Properties of the set mapping to $\Omega(g(\xi))$}
\label{sub:OmegaPrime}

Suppose $\xi \in \Phi$ with $\zeta_1(\xi) > 0$ (equivalently $\psi(\xi) < 0$).
Then $f_\xi^2(T)$ lies to the right of $E^s(X)$ by Proposition \ref{pr:psi}.
Thus $f_\xi^3(T)$ lies to the left of $E^s(X)$ (because $\lambda_R^u < 0$).
Now let $Q$ denote the intersection of $E^u(X)$ with the line through $f_\xi^3(T)$ and parallel to $E^s(X)$, see Fig.~\ref{fig:igRN_X}-b.
Then let $\Omega'(\xi)$ be the filled triangle with vertices $f_\xi(T)$, $f_\xi^3(T)$, and $Q$.

\begin{lemma}
Let $\xi \in \Phi$ with $\zeta_1(\xi) > 0$.
Then
\begin{romanlist}
\item
\label{it:InPi}
$\Omega'(\xi) \subset \Pi_\xi$,
\item
\label{it:ImageDisjoint}
$\Omega'(\xi) \cap f_\xi \left( \Omega'(\xi) \right) = \varnothing$,
\item
\label{it:ImageInRight}
$f_\xi \left( \Omega'(\xi) \right) \subset \left\{ (x,y) \in \mathbb{R}^2 \,\big|\, x > 0 \right\}$,
\item
\label{it:MapsToOmega}
$h_\xi \left( \Omega'(\xi) \right) = \Omega(g(\xi))$,
\item
\label{it:InOmega}
and if $\zeta_0(\xi) > 0$ then $\Omega'(\xi) \subset \Omega(\xi)$.
\end{romanlist}
\label{le:OmegaPrime}
\end{lemma}

\begin{proof}
Let $\Xi_R = \left\{ (x,y) \in \mathbb{R}^2 \,\big|\, x > 0 \right\}$ denote the open right half-plane
and let $\Psi$ be the triangle with vertices $X$, $f_\xi(T)$, and $V$.
We now prove parts \ref{it:InPi}--\ref{it:InOmega} in order.
\begin{romanlist}
\item
Observe $f_\xi(X) = X \in \Xi_R$, thus $X \in \Pi_\xi$ by \eqref{eq:Pi}.
Similarly $f_\xi(V) \in \Xi_R$, thus $V \in \Pi_\xi$.
Also $f_\xi^2(T) \in \Xi_R$, thus $f_\xi(T) \in \Pi_\xi$.
That is, all vertices of $\Psi$ belong to $\Pi_\xi$,
thus $\Psi \subset \Pi_\xi$ because these sets are convex.

From \eqref{eq:T} and \eqref{eq:f2T} we find that
the slope of the line through $T$ and $f_\xi^2(T)$ is
$\frac{-\delta_L}{\tau_L - \lambda_R^s}$, which is negative,
thus $f_\xi^2(T)$ lies to the left of $T$.
Consequently $f_\xi^3(T)$ lies above $f_\xi(T)$.
Also $f_\xi(T)$ lies above $V$ because
\begin{equation}
f_\xi(T)_2 - V_2 = \frac{1 - \delta_R}{\left( 1 - \lambda_R^s \right) \left( 1 - \frac{1}{\lambda_R^u} \right)} > 0.
\nonumber
\end{equation}
Therefore $f_\xi^3(T) \in \Psi$.
Thus $\Omega'(\xi) \subset \Psi \subset \Pi_\xi$.
\item
Observe $f_\xi(\Psi)$ is the quadrilateral with vertices
$X$, $f_\xi(V)$, $f_\xi^2(T)$, and $T$.
Thus $\Psi$ and $f_\xi(\Psi)$ intersect only at $X$.
But $\Omega'(\xi) \subset \Psi$ does not contain $X$,
thus $\Omega'(\xi) \cap f_\xi \left( \Omega'(\xi) \right) = \varnothing$.
\item
The left-most point of $f_\xi(\Psi)$ is $X \in \Xi_R$,
thus $f_\xi \left( \Omega'(\xi) \right) \subset f_\xi(\Psi) \subset \Xi_R$.
\item
For the map $f_\xi^2$, the fixed point $X$ is a saddle with positive eigenvalues.
Thus its unstable manifold has two dynamically independent branches.
The branch that emanates to the left has its first and second kinks at $f_\xi(T)$ and $f_\xi^3(T)$.
Let $\mathcal{B}$ denote this branch up to the second kink,
that is $\mathcal{B}$ is the union of the line segments $\roverline{X f_\xi(T)}$ and $\roverline{f_\xi(T) f_\xi^3(T)}$.

By the conjugacy relation \eqref{eq:conjugacy},
$h_\xi(\mathcal{B})$ is part of one branch of the unstable manifold of the analogous fixed point of $f_{g(\xi)}$.
Since $h_\xi$ flips points across the switching line \eqref{eq:oppositeSigns},
$h_\xi(\mathcal{B})$ is part of the unstable manifold
of $Y$ (for the map $f_{g(\xi)}$).
This branch has its first and second kinks at $D$ and $f_{g(\xi)}(D)$,
thus $h_\xi(\mathcal{B})$ is the union of the line segments $\roverline{Y D}$ and $\roverline{D f_{g(\xi)}(D)}$.
By similar reasoning $Q$ maps under $h_\xi$ to the point $B$ of $f_{g(\xi)}$.
This verifies part \ref{it:MapsToOmega}.
\item
The first components of $T$ and $D$
are $T_1 = \frac{1}{1 - \lambda_R^s}$ and $D_1 = \frac{1}{1 - \lambda_L^s}$.
Observe $0 < T_1 < D_1$, thus $T$ lies between $(0,0)$ and $D$.
By iterating these under $f_{R,\xi}$ we have that
$f_\xi(T)$ lies on the line segment connecting $(1,0)$ and $f_\xi(D)$.

Now suppose $\zeta_0(\xi) > 0$.
Then $f_\xi(T) \in \Omega(\xi)$ because $(1,0) \in \Omega(\xi)$, $f_\xi(D) \in \Omega(\xi)$, and $\Omega(\xi)$ is convex.
Moreover, $f_\xi^3(T) \in \Omega(\xi)$ because $\Omega(\xi)$ is forward invariant (Proposition \ref{pr:Omega}).
Also $X \in \Omega(\xi)$ by Lemma \ref{le:LambdaInOmega}.
Thus the triangle with vertices $f_\xi(T)$, $f_\xi^3(T)$, and $X$ is contained in $\Omega(\xi)$
(again by the convexity of $\Omega(\xi)$).
This triangle contains $\Omega'(\xi)$, thus $\Omega'(\xi) \subset \Omega(\xi)$ as required.
\end{romanlist}
\end{proof}


\subsection{Arguments leading to a proof of Theorem \ref{th:affinelyConjugate}}
\label{sub:proofByInduction}



\begin{proof}[Proof of Theorem \ref{th:affinelyConjugate}]
Let $I_n = \{ 0,1,\ldots, 2^n-1 \}$.
We use induction on $n$ to prove Theorem \ref{th:affinelyConjugate}
and show that
\begin{equation}
\text{if $\zeta_0(\xi) > 0$ then $S_i \subset \Omega(\xi)$ for all $i \in I_n$}.
\label{eq:affinelyConjugateProof0}
\end{equation}
With $n=0$ the statements in Theorem \ref{th:affinelyConjugate} are true trivially with $S_0 = \Lambda(\xi)$.
Also \eqref{eq:affinelyConjugateProof0} is true because $\zeta_0(\xi) > 0$
(since $\xi \in \cR_0$) and $S_0 \subset \Omega(\xi)$ by Lemma \ref{le:LambdaInOmega}.

Now suppose the result is true for some $n \ge 0$;
it remains for us to verify the result for $n+1$.
Choose any $\xi \in \cR_{n+1}$.
Then $g(\xi) \in \cR_n$ by Lemma \ref{le:gForwards}.
By the induction hypothesis applied to the point $g(\xi)$,
we have $g^{n+1}(\xi) \in \cR_0$ and there exist mutually disjoint sets
$\tilde{S}_0, \tilde{S}_1, \ldots, \tilde{S}_{2^n-1} \subset \mathbb{R}^2$
with $f_{g(\xi)} \left( \tilde{S}_i \right) = \tilde{S}_{(i+1) \,{\rm mod}\, 2^n}$ and 
\begin{equation}
f_{g(\xi)}^{2^n} \big|_{\tilde{S}_i} ~\text{is affinely conjugate to}~ f_{g^{n+1}(\xi)} \big|_{\Lambda(g^{n+1}(\xi))}
\label{eq:affinelyConjugateProof10}
\end{equation}
for all $i \in I_n$.
Also $\zeta_0(g(\xi)) > 0$ by Lemma \ref{le:zeta0gxi},
thus by \eqref{eq:affinelyConjugateProof0} the induction hypothesis also gives 
$\tilde{S}_i \subset \Omega(g(\xi))$ for all $i \in I_n$.

Let $S_{2i} = h_\xi^{-1} \left( \tilde{S}_i \right)$ for each $i \in I_n$
(these sets are mutually disjoint because $h_\xi$ is a homeomorphism).
Let $S_{2i+1} = f_\xi(S_{2i})$ for each $i \in I_n$
(these sets are mutually disjoint because $f_\xi$ is a homeomorphism).
For any $i,j \in I_n$ we have $S_{2i} \subset \Omega'(\xi)$ by Lemma \ref{le:OmegaPrime}\ref{it:MapsToOmega}
and $S_{2j+1} \cap \Omega'(\xi) = \varnothing$ by Lemma \ref{le:OmegaPrime}\ref{it:ImageDisjoint},
so $S_{2i} \cap S_{2j+1} = \varnothing$.
Therefore the sets $S_0, S_1, \ldots, S_{2^{n+1}-1}$ are mutually disjoint.

For each $i \in I_n$, $S_{2i} \subset \Pi_\xi$ by Lemma \ref{le:OmegaPrime}\ref{it:InPi}, so
\begin{equation}
f_\xi^2 \big|_{S_{2i}} ~\text{is affinely conjugate to}~ f_{g(\xi)} \big|_{\tilde{S}_i}
\label{eq:affinelyConjugateProof20}
\end{equation}
by Proposition \ref{pr:conjugacy}.
Also $f_\xi^2(S_{2i}) = S_{2i+2 \,{\rm mod}\, 2^{n+1}}$,
so $f_\xi(S_{2i+1}) = f_{R,\xi}(S_{2i+1}) = S_{2i+2 \,{\rm mod}\, 2^{n+1}}$
using also Lemma \ref{le:OmegaPrime}\ref{it:ImageInRight}.
Thus
\begin{equation}
f_\xi^2 \big|_{S_{2i+1}} ~\text{is affinely conjugate to}~ f_\xi^2 \big|_{S_{2i}}
\nonumber
\end{equation}
using $f_{R,\xi}$ as the affine transformation.
By further use of \eqref{eq:conjugacy} we have that 
$f_\xi^{2^{n+1}} \big|_{S_{2i}}$ and $f_\xi^{2^{n+1}} \big|_{S_{2i+1}}$
are affinely conjugate to $f_{g(\xi)}^{2^n} \big|_{\tilde{S}_i}$,
thus also to $f_{g^{n+1}(\xi)} \big|_{\Lambda(g^{n+1}(\xi))}$ by \eqref{eq:affinelyConjugateProof10}
(this verifies \eqref{eq:affinelyConjugate} for $n+1$).

The induction hypothesis also implies
\begin{equation}
\bigcup_{i=0}^{2^n-1} \tilde{S}_i = {\rm cl} \left( W^u(\gamma_n) \right),
\label{eq:affinelyConjugateProof40}
\end{equation}
where $\gamma_n$ is a periodic solution of $f_{g(\xi)}$ with symbolic itinerary $\cF^n(R)$.
By \eqref{eq:conjugacy},
$h_\xi^{-1}(\gamma_n)$ is a periodic solution of $f_\xi^2$.
Since $h_\xi$ flips the left and right half-planes, see \eqref{eq:oppositeSigns},
the symbolic itinerary of $h_\xi^{-1}(\gamma_n)$ is obtained by swapping $L$ and $R$'s in $\cF^n(R)$.
Then $\gamma_{n+1} = h_\xi^{-1}(\gamma_n) \cup f_\xi \left( h_\xi^{-1}(\gamma_n) \right)$
is a periodic solution of $f_\xi$ and
since $f_\xi \left( h_\xi^{-1}(\gamma_n) \right)$ is contained in the right half-plane
(Lemma \ref{le:OmegaPrime}\ref{it:ImageInRight})
its symbolic itinerary is obtained
by further replacing each $L$ with $LR$ and each $R$ with $RR$,
hence $\gamma_{n+1}$ has symbolic itinerary $\cF^{n+1}(R)$.
Also by \eqref{eq:affinelyConjugateProof20} and \eqref{eq:affinelyConjugateProof40},
\begin{equation}
\bigcup_{i=0}^{2^{n+1}-1} S_i = {\rm cl} \left( W^u(\gamma_{n+1}) \right),
\nonumber
\end{equation}
which verifies \eqref{eq:Siunion} for $n+1$.
Finally, if $\zeta_0(\xi) > 0$ then for all $i \in I_n$ we have
$S_{2i} \subset \Omega(\xi)$ by Lemma \ref{le:OmegaPrime}\ref{it:InOmega}
and $S_{2i+1} \subset \Omega(\xi)$ because $\Omega(\xi)$ is forward invariant
verifying \eqref{eq:affinelyConjugateProof0} for $n+1$.
\end{proof}

\section{Discussion}
\label{sec:conc}
\setcounter{equation}{0}


In this paper we have shown how part of the parameter space of \eqref{eq:f}
naturally divides into regions $\cR_0, \cR_1, \ldots$.
As demonstrated by Theorem \ref{th:affinelyConjugate},
renormalisation enables us to describe the dynamics in each $\cR_n$ with $n \ge 1$ based on knowledge of the dynamics in $\cR_0$.
Theorem \ref{th:R0} describes the dynamics in $\cR_0$, but is incomplete.
It remains to show the attractor $\Lambda$ is unique and satisfies stronger notions of chaos throughout $\cR_0$.
Also we would like to extend the results to high-dimensional maps.

Finally we comment on the analogy of Feigenbaum's constant for our renormalisation
by looking at the rate at which the regions $\cR_n$ converge to the fixed point $\xi^*$.
The $4 \times 4$ Jacobian matrix $\rD g(\xi^*)$ has exactly one unstable eigenvalue: $2$.
It follows that the diameter of $\cR_n$ divided by the diameter of $\cR_{n+1}$ tends, as $n \to \infty$, to the constant $2$.

\section*{Acknowledgements}

The authors were supported by Marsden Fund contract MAU1809,
managed by Royal Society Te Ap\={a}rangi.

\appendix

\end{document}